\documentclass[a4paper,10pt,envcountsect,orivec]{llncs}

\usepackage{times}
\usepackage{pslatex}
\usepackage{amssymb}
\usepackage{multirow}
\usepackage{suffix}
\usepackage{mathrsfs}
\usepackage{comment}

\usepackage{etex}

\usepackage{xy}
\xyoption{all}


\usepackage[mathscr]{euscript}
\usepackage{mathptmx}
\usepackage{courier}
\usepackage{tikz} 
\usetikzlibrary{arrows,decorations,backgrounds,positioning,fit,petri}
\RequirePackage{ifthen}
\usepackage{rotating}

\usepackage{url}
\usepackage{epic,eepic}
\usepackage{amssymb}
\usepackage{xcolor}
\usepackage{citesort}
\RequirePackage[final]{listings}
\usepackage{amsmath}
\usepackage{latexsym}
\usepackage{xspace}
\usepackage{prooftree}

\usepackage{wrapfig}

\usepackage{space}

\numberwithin{equation}{section}

\newif\iflong\longfalse

\newif\ifpopllong\popllongfalse

\popllongtrue

\newif\ificalp\icalpfalse

\icalptrue

\usepackage{xspace}
\usepackage{listings}
\usepackage{textcomp}

\newcommand{\keyword}[1]{\textsf{#1}\xspace}

\newcommand{\In}{\keyword{in}}

\newcommand{\register}[1]{\keyword r_{#1}}

\newcommand{\lockbit}[1]{\textbf{#1}}

\newcommand{\subs}[2]{[{#1}/{#2}]}

\newcommand{\pad}{\;\;}
\newcommand{\Space}[1]{\pad{#1}\pad}
\newcommand{\grmeq}{{\Space{::=}}}

\newcommand{\grmor}{{\;\text{\large$\mid$}\;}}

\newcommand{\mycaption}[1]{
  \iflong 
  \rule{\textwidth}{1pt} 
  \vspace{-5.5ex} 
  \fi
  \caption{#1}
  \iflong 
  \vspace{-1.5ex}
  \rule{\textwidth}{.5pt}
  \vspace{-2ex} 
  \fi
}

\newcommand{\myparagraph}[1]{\paragraph{\bf{#1}}}

\newcommand{\ENCan}[1]{\langle #1 \rangle}
{}{}
{}{}
{}{}
{}{}
{}{}
{}{}
{}{}
{}{}
{}{}

\newcommand{\defk}{\keyword{def}}

\newcommand{\PAR}{\mid}        

\newcommand{\OR}{\mathrel{\&}}

\newcommand{\participant}[1]{\mathtt{#1}}

\newcommand{\SEND}[2]{#1 ! #2}

\newcommand{\RECV}[2]{#1 ?#2}

\newcommand{\typevar}{\mathsf{t}}

\newcommand{\osred}{\rightarrow}                
\newcommand{\R}{\osred}
\newcommand{\msred}{\osred^\ast}     
\newcommand{\RR}{\msred}

\newcommand{\TRANS}[1]{\xrightarrow{#1}}

\newcommand{\TRANSS}[1]{{\xrightarrow{\raisebox{-.3ex}[0pt][0pt]{\scriptsize $#1$}}}}

\newcommand{\rulename}[1]{[\text{\sc #1}]\xspace}

\newcommand{\branch}{\&}

\newcommand{\typeconst}[1]{\keyword{#1}}

\newcommand{\End}{\typeconst{end}}

\def\recv#1{?\,\ENCan{#1}}
\def\send#1{!\,\ENCan{#1}}

\newcommand{\TO}[2]{\participant{#1}\to\participant{#2}}
\newcommand{\TOS}[2]{\participant{#1}\rightsquigarrow\participant{#2}}

\newcommand{\GBRA}[6]{\TO{#1}{#2}\colon \ASET{#3. #5}_{#6}}
\newcommand{\GBRAS}[7]{\TOS{#1}{#2}\colon {#3}\; \ASET{#4. #6}_{#7}}

\newcommand{\LSENDK}[5]{#1 !{\ASET{#2.#4}_{#5}}}
\newcommand{\LRECVK}[5]{#1 ?{\ASET{#2.#4}_{#5}}}

\newcommand{\LLPSEND}[3]{\send{#1,#2\ENCan{#3}}}
\newcommand{\LLPRECV}[3]{\recv{#1,#2\ENCan{#3}}}

\newcommand{\rcdt}{\{l_i\colon T_i\}_{i\in I}}
\newcommand{\brancht}[1][\alpha]{\branch\rcdt}

\newcommand{\AT}[2]{#1\colon\! #2}

\newcommand{\co}{\overline}
\newcommand{\OL}{\co}

\newcommand{\ASET}[1]{\{ #1 \}}
\newcommand{\VEC}{\tilde}

\newif\ifny\nytrue
\nyfalse

\newif\ifvv\vvtrue
\vvfalse

\newif\ifkohei\koheitrue
\koheifalse

\newif\ifmarco\marcotrue
\marcofalse

\def\SETRES{\!\upharpoonright\!}
\def\fps@figure{tp}      
\def\fps@table{tp}
\def\topfraction{.98}      
\def\floatpagefraction{.9}
\newcommand{\NUL}{\epsilon}

 \newcommand{\WB}{\approx}

\newcommand{\enewchan}[2]{\ensuremath{\mathtt{newChan}\ \AT{#1}{#2}}}

\newcommand{\MRED}[1][]{%
  \ensuremath{%
    \ifthenelse{\equal{#1}{}}{%
      \rightarrow\!\!\!\!\rightarrow%
    }{%
      \rightarrow\!\!\!\!\rightarrow_{#1}%
    }%
  }%
}

\newcommand{\RED}[1][]{%
  \ensuremath{%
    \ifthenelse{\equal{#1}{}}{%
      \longrightarrow%
    }{%
      \longrightarrow_{#1}%
    }%
  }%
}

\newcommand{\infer}[2]{\frac{\displaystyle{ #1 }}
  {\rule{0pt}{2.2ex}\displaystyle{ #2 }}}

\newcommand{\LL}%
{{\mathsf{L}\mathsf{L}\mathsf{S}}}

\newcommand{\LLC}%
{{\mathsf{L}\mathsf{L}\mathsf{S}}^{\mathsf{C}}}

\newcommand{\LLA}%
{{\mathsf{L}\mathsf{L}\mathsf{S}}^{\mathsf{A}}}

\newcommand{\newlinchan}[3]
{\ensuremath{\mathtt{newCont}\ \AT{#1}{#2}\ \mathtt{in}\ #3}}

\newcommand{\newchan}[3]
{\enewchan{#1}{#2}\,;\,#3}

\newcommand{\SIL}%
{{\mathsf{S}\mathsf{I}\mathsf{L}}\xspace}

\newcommand{\SILc}%
{{\mathsf{S}\mathsf{I}\mathsf{L}^{\text{\scriptsize C}}}\xspace}

\newcommand{\SILC}%
{{\mathsf{S}\mathsf{I}\mathsf{L}^{\text{\scriptsize C}}}\xspace}

\newcommand{\SILa}%
{{\mathsf{S}\mathsf{I}\mathsf{L}^{\text{\scriptsize A}}}\xspace}

\newcommand{\SILA}%
{{\mathsf{S}\mathsf{I}\mathsf{L}^{\text{\scriptsize A}}}\xspace}

{\end{array}%
 \right.}

\newcommand{\QED}{\hfill $\square$}

\newcommand{\VECw}[1]{{\widetilde{#1}}} 

\newcommand{\causes}[1]{\searrow}

\newif\ifdm\dmtrue

\newif\ifdmr
\dmrfalse

\newcommand{\PSet}{\!\mathscr{P}\!}

\newcommand{\xx}{\ensuremath{\mathbf{x}}}

\newcommand{\wv}{\ensuremath{\VEC{w}}}

\newcommand{\XXV}{\ensuremath{\VEC{\XX}}}

\newcommand{\G}{\ensuremath{G}}

\newcommand{\proj}[1]{\ensuremath{\upharpoonright #1}}
\newcommand{\projrel}[1]{\ensuremath{\text{\textlbrackdbl} #1 \text{\textrbrackdbl}}}

\newcounter{analphabet}
{\rm%
\begin{list}%
{\arabic{analphabet}. }%
{\usecounter{analphabet}%
 \addtolength{\labelwidth}{5mm}%
\addtolength{\leftmargin}{-2mm}%
\setlength{\rightmargin}{0pt}%
\setlength{\itemsep}{0mm}%
\setlength{\parsep}{0pt}}}%
{\end{list}}

\newcommand{\mergeop}{\ensuremath{\bowtie}}
\newcommand{\mergecup}{\ensuremath{\sqcup}}

\newcommand{\chanset}{C}
 
\newcommand{\RS}{\mathit{RS}}

\newcommand{\TR}{\mathit{Tr}} 
\newcommand{\MEM}{\mathit{bound}}

\newcommand{\Alice}{\ensuremath{\mathtt{A}}}
\newcommand{\Bob}{\ensuremath{\mathtt{B}}}
\newcommand{\Carol}{\ensuremath{\mathtt{C}}}

\newcommand{\p}{\ensuremath{\participant{p}}}
\newcommand{\q}{\ensuremath{\participant{q}}}
\newcommand{\pp}{\ensuremath{\participant{p}}}
\newcommand{\qq}{\ensuremath{\participant{q}}}

\newcommand{\action}{\ell}
\newcommand{\act}{\mathit{act}}
\newcommand{\sub}{\mathit{subj}}
\newcommand{\subj}{\mathit{subj}}

\newcommand{\A}{\ensuremath{\mathscr{A}}}
\newcommand{\LT}{\ensuremath{\mathscr{T}}}
\newcommand{\X}{\ensuremath{\participant{X}}}
\newcommand{\XX}{\ensuremath{\mathbf{X}}}

\newcommand{\GG}{\ensuremath{\mathbf{G}}}
\newcommand{\GV}{\ensuremath{\VECw{G}}}
\newcommand{\TT}{\ensuremath{\mathbf{T}}}
\newcommand{\TV}{\ensuremath{\VECw{T}}}

\newcommand{\TTV}{\ensuremath{\VECw{\TT}}}

\newcommand{\equivGV}{\equiv_\GV}
\newcommand{\equivTV}{\equiv_\TV}

\newcommand{\Rcv}{\ensuremath{\mathit{Rcv}}}
\newcommand{\ASend}{\ensuremath{\mathit{ASend}}}

\newcommand{\Asender}{\ASend}

\newcommand{\tsrule}[1]{{\text{\scriptsize{$\lfloor$\scriptsize\sc{#1}$\rfloor$}}}}

\newcommand{\T}{\ensuremath{T}}

\newcommand{\s}{\ensuremath{s}}

\DeclareSymbolFont{bbsymbol}{U}{bbold}{m}{n}
\DeclareMathSymbol{\bbsemicolon}{\mathbin}{bbsymbol}{"3B}

\newcommand{\U}{U}

\newcommand{\ASigma}{ \mathbb{A}}
\newcommand{\PN}{ \mathbb{P}}

\newcommand{\MSACore}%
{CMSA\xspace}
\newcommand{\MSA}%
{MSA\xspace}

\newcommand{\isnow}{\leftarrow}

\begin{document}
\title{
\vspace*{-2mm}
Multiparty Compatibility in Communicating Automata:\\
Characterisation and Synthesis of Global Session Types
\vspace*{-3mm}
}


 \author{
 Pierre-Malo Deni\'elou$^1$ 
 \and Nobuko Yoshida$^2$ 
\vspace*{-1mm}
 }

 \institute{
 $^1$Royal Holloway, University of London \quad \quad 
 $^2$Imperial College London
 }




\maketitle 

\thispagestyle{plain}
\pagestyle{plain}


\begin{abstract}


Multiparty session types are a type system that can ensure the safety
and liveness of distributed peers via the global specification of their
interactions.  
To construct a global specification 
from a set of distributed uncontrolled behaviours, 
this paper explores the problem of fully characterising multiparty session
types in terms of communicating automata. We equip global and local
session types with labelled transition systems (LTSs) that
faithfully represent asynchronous communications through unbounded buffered
channels. 
Using the equivalence between the two LTSs, we identify 
\ificalp
a class of 
communicating automata that exactly correspond to the projected
local types.
\fi
We exhibit an algorithm to synthesise a global type from a collection of
communicating automata. The key property of our findings is the notion of
{\em multiparty compatibility} which non-trivially extends the duality
condition for binary session types. 


\end{abstract}

\vspace*{-7mm}
\section{Introduction}
\label{sec:introduction}
Over the last decade, {\em session types} \cite{HKV98,THK} 
have been studied as 
data types or functional types for communications and distributed systems. 
A recent discovery by~\cite{CF10,Wadler2012}, which establishes a Curry-Howard
isomorphism between binary session types and linear logics, confirms
that 
session types and the notion of duality between type constructs have
canonical meanings. 
On the practical side, 
multiparty session types 
\cite{CHY07,BettiniCDLDY08} were proposed as 
a major generalisation of binary
session types. It can 
enforce communication safety and deadlock-freedom for more than two peers
thanks to a choreographic 
specification (called {\em global type}) 
of the interaction.  
Global types are projected to end-point types (called {\em local
  types}), against which processes can be statically type-checked and
verified to behave correctly. 

The motivation of this paper comes from our practical experiences that, in
many situations, even where we start from the end-point projections of
a choreography, 
we need to reconstruct a global type from distributed 
specifications. End-point specifications are usually 
available, either through inference from the control flow, 
or through existing service interfaces, and always in
forms akin to individual communicating finite state machines. 
If one knows the 
precise conditions under which a global type can be constructed (i.e.~the conditions of {\em synthesis}), 
not only the global safety property 
which multiparty session types 
ensure is guaranteed, but also the generated global type 
can be used as a refinement and be integrated within
the distributed system development life-cycle  
(see \S~\ref{sec:related} for applications \cite{savara,ooi}). 
This paper attempts to give the synthesis condition 
as a sound and complete characterisation of multiparty session types 
with respect to Communicating Finite State Machines (CFSMs) \cite{Brand83}.
CFSMs have been a well-studied formalism 
for analysing distributed safety
properties and are widely present in industry tools. They 
can been seen as 
generalised end-point specifications,   
therefore, an
excellent target for a common comparison ground and for synthesis. 
As explained below, to identify a complete set of CFSMs for synthesis, 
we first need to answer a   
question -- {\em what is the canonical duality notion in multiparty session types?} 
\\[1mm]
{\bf Characterisation of binary session types as communicating
automata} 
The subclass which fully characterises {\em binary session types} 
was actually proposed by Gouda, Manning and Yu in
1984~\cite{GMY84} in a pure communicating automata 
context.\footnote{Villard \cite{Villard11} independently found this
  subset in the context of channel contracts \cite{FAHHHLL06}.} 
Consider a simple business protocol between 
a Buyer and a Seller from the Buyer's viewpoint: Buyer sends the 
title of a book, Seller answers with a quote. If
Buyer is satisfied by the quote, then he sends his address and
Seller sends back the delivery date; otherwise it retries the
same conversation. This can be described by the following
session type:\vspace{-2ex}

\begin{equation}\vspace{-1ex}
\label{AAA}
\mu \typevar.
!\, \typeconst{title};\
?\typeconst{quote};\
!\{\ \mathsf{ok}:!\,\typeconst{addrs};
?\typeconst{date};\End, \ 
\quad\mathsf{retry}:\typevar\ \}
\end{equation}
where the operator $!\,\typeconst{title}$ denotes an
output of the title, whereas
$?\mathsf{quote}$ denotes an input of a quote. 
The output choice features the two options
$\mathsf{ok}$ and $\mathsf{retry}$ and  $;$ denotes sequencing.
\End represents the termination of the session, and $\mu \typevar$ is 
recursion.

The simplicity and tractability of 
binary sessions come from the notion of {\em duality} in
interactions \cite{GirardJY:linlog}.  
The interaction pattern of the Seller
is fully given as the dual of the type in (\ref{AAA}) (exchanging input $!$ 
and output $?$ in the original type).
When composing two parties, we only have to check
they have mutually dual types, and the resulting communication is 
guaranteed to be deadlock-free. 
Essentially the same characterisation is given in communicating
automata. Buyer and Seller's session types are represented by the
following
two machines. 
\vspace{-2ex}
\[
\rightarrow\!\xymatrix@C=1.7em@R=4.ex{
  *++[o][F]{}\ar[r]_{!\text{title}} & *++[o][F]{}\ar[r]_{?\text{quote}} 
  & *++[o][F]{}\ar@/_2ex/[ll]_{?\text{retry}}\ar[r]_{?\text{ok}} &
  *++[o][F]{}\ar[r]_{!\text{addrs}}& *++[o][F]{}\ar[r]_{?\text{date}}
    & *++[o][F=]{}}
\quad 
\rightarrow\!\xymatrix@C=1.7em@R=4.ex{
  *++[o][F]{}\ar[r]_{?\text{title}} & *++[o][F]{}\ar[r]_{!\text{quote}} 
  & *++[o][F]{}\ar@/_2ex/[ll]_{!\text{retry}}\ar[r]_{!\text{ok}} &
    *++[o][F]{}\ar[r]_{?\text{addrs}}& *++[o][F]{}\ar[r]_{!\text{date}}
  & *++[o][F=]{}
}
\vspace{-1ex}
\]
We can observe that these CFSMs satisfy three conditions.  First, the
communications are {\em deterministic}: messages that are part of the same
choice, ok and retry here, are distinct. Secondly, there
is no mixed state (each state has either only sending actions or only
receiving actions). Third, these two machines have {\em compatible} traces
(i.e. dual): the Seller machine can be defined by exchanging sending to
receiving actions and vice versa. Breaking one of these conditions allows
deadlock situations and breaking one of the first two conditions makes the
compatibility checking undecidable~\cite{GMY84}. 
%
\\[1mm]
{\bf Multiparty compatibility}
This notion of duality is no longer effective in multiparty
communications, where the whole conversation cannot be reconstructed
from only a single behaviour. 
To bypass the gap between binary and multiparty, 
we take the {\em synthesis} approach, that is to find conditions which allow
a global choreography to be built from the local machine behaviour. 
Instead of directly trying to decide whether the communications of a system
will satisfy safety (which is undecidable in the general
case), inferring a global type guarantees the safety as a direct consequence.


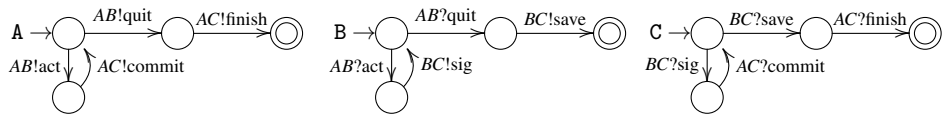
\begin{figure}
\[\begin{array}{@{}ll@{}}
  \begin{array}{ll@{\qquad}ll} 
\mathtt{A}\rightarrow\hspace{-1.3cm}
\xymatrix@C=3.2em@R=3ex{
   & *++[o][F]{}\ar[r]^{AB!\text{quit}}\ar[d]_{AB!\text{act}}
   & *++[o][F]{}\ar[r]^{AC!\text{finish}}
   & *++[o][F=]{}  \\
   & *++[o][F]{}\ar@/_2ex/[u]_{AC!\text{commit}}  \\
 } &
\quad
\mathtt{B}\rightarrow\hspace{-1.3cm}
  \xymatrix@C=3.2em@R=3ex{
  %
   &*++[o][F]{}\ar[r]^{AB?\text{quit}}\ar[d]_{AB?\text{act}}
   & *++[o][F]{}\ar[r]^{BC!\text{save}}
   & *++[o][F=]{}  \\
   & *++[o][F]{}\ar@/_2ex/[u]_{BC!\text{sig}} \\
 }
\quad
\mathtt{C}\rightarrow\hspace{-1.3cm}
\xymatrix@C=3.2em@R=3ex{
   & 
   *++[o][F]{}\ar[r]^{BC?\text{save}}\ar[d]_{BC?\text{sig}}
   & *++[o][F]{}\ar[r]^{AC?\text{finish}}
   & *++[o][F=]{}  \\
   & *++[o][F]{}\ar@/_2ex/[u]_{AC?\text{commit}}  \\
 }
\end{array}
\end{array}
\]
\caption{Commit example: CFSMs}
\label{fig:commit}
\end{figure}

\noindent We give a simple example to illustrate
the problem.  The commit protocol in Figure~\ref{fig:commit} 
involves three machines: Alice $\Alice$, 
Bob $\Bob$ and Carol $\Carol$. $\Alice$ orders $\Bob$ to act
or quit. If act is sent,  $\Bob$ sends a signal to $\Carol$, and 
$\Alice$ sends a commitment to $\Carol$ and continues.
Otherwise $\Bob$ informs $\Carol$ to save the data 
and $\Alice$ gives the final notification 
to $\Carol$ to terminate the protocol. 
%

This paper presents a
decidable notion of {\em multiparty compatibility} as a generalisation of
duality of binary sessions, which in turns characterises 
a synthesis condition. 
The idea is to check the duality between each automaton and the rest,
up to the internal communications (1-bounded executions 
in the terminology of CFSMs, see \S~\ref{sec:cfsm}) that the other machines will
independently perform. For example, in Figure~\ref{fig:commit}, 
to check the compatibility of trace 
$\mathit{BC}?\text{sig}$
$\mathit{AC}?\text{commit}$ in $\Carol$, 
we execute the internal communications between $\Alice$ and $\Bob$
such that $\mathit{AB}!\text{act}\cdot\mathit{AB}?\text{act}$ 
and observes the dual trace $\mathit{BC}!\text{sig}\cdot\mathit{AC}!\text{commit}$ from $\Bob$ and $\Alice$. 
If this extended duality is valid for all the
machines from any 1-bounded reachable state, 
then they satisfy multiparty compatibility and 
can build a well-formed global choreography. 
\\[1mm]
{\bf Contributions and Outline \ }
Section~\ref{sec:globallocal} defines new labelled transition systems 
for global and local types that represent the abstract observable behaviour
of typed processes. We prove that a global type behaves 
exactly as its projected local types, and the same result
between a single local type and its CFSMs interpretation. 
These correspondences are the key to prove the main theorems.  
Section~\ref{sec:cmsa} 
defines 
  {\em multiparty compatibility}, studies its safety and liveness properties, 
  gives an algorithm for the synthesis of global types from CFSMs, 
  and proves the soundness and completeness results between global types 
and CFSMs.  
%
%
%
Section~\ref{sec:related} discusses related work and concludes. 
The full proofs can be found in Appendix. 

In Appendix \ref{sec:graph}, we also extend our result to {\em generalised
    multiparty session types}, a recent class of multiparty session
types~\cite{DY12} with graph-like control flow and parallelism. 
The same multiparty compatibility as in \S~\ref{sec:cmsa} can be used without
modification, although well-formedness condition need to be generalised.
The synthesis algorithm relies on Petri net intermediate
representations \cite{cortadella1998deriving} and 1-bounded behavioural exploration.
Our result is applicable to generate a core part of Choreography BPMN 2.0
specification \cite{BPMN} from CFSMs.

\section{Communicating Finite State Machines}
\label{sec:cfsm}
This section starts from some preliminary notations (following~\cite{CeceF05}). 
\label{def:sub}
$\NUL$ is the empty word. $\ASigma$ is a finite
alphabet and $\ASigma^\ast$ is the set of all finite words over $\ASigma$. $|x|$
is the length of a word $x$ and $x.y$ or $xy$ the concatenation of two words $x$
and $y$. 
Let $\PSet$ be a set of {\em participants} fixed throughout the
paper: $\PSet \subseteq
\{\Alice,\Bob,\Carol, \ldots,\p, \q, \dots\}$. 

\begin{definition}[CFSM]\label{def:cfsm2}\rm 
A communicating finite state machine is a finite transition
system given by a 5-tuple $M=(Q,\chanset,q_0,\ASigma,\delta)$ where (1) $Q$ is 
a finite set of {\em states};  (2) 
$\chanset=\ASET{\p\q\in\PSet^2\PAR\p \not = \q}$ 
is a set of channels; (3) $q_0\in Q$ is an initial state; 
(4) $\ASigma$ is a finite {\em alphabet} of messages, and 
(5) $\delta\ \subseteq \ Q \times 
(\chanset \times \ASET{!,?} \times \ASigma) \times Q$ is a finite 
set of {\em transitions}. 
\end{definition}
\noindent 
In transitions, 
$\p\q!a$ denotes the {\em sending} action of $a$ from process $\p$ 
to process $\q$, and 
$\p\q ?a$ denotes the {\em receiving} action of $a$ from $\p$ by $\q$.
$\action,\action'$ range over actions and 
we define the {\em subject} of an action $\ell$ as the principal in charge
of it: $\sub(\p\q!a)=\sub(\q\p?a)=\p$.


A state $q\in Q$ whose outgoing transitions are all labelled with sending
(resp. receiving) actions is called a {\em sending} (resp.~{\em receiving})
state. A state $q\in Q$ which does not have any outgoing transition is
called {\em final}.  If $q$ has both sending and receiving outgoing
transitions, $q$ is called {\em mixed}. 
We say $q$ is {\em directed} if it
contains only sending (resp. receiving) actions to (resp. from) the same
participant.
A {\em path} in $M$ is a finite sequence of $q_0,\ldots,q_{n}$ ($n\geq 1$) such
that $(q_{i},\ell,q_{i+1})\in \delta$ ($0\leq i \leq n-1$), and  
we write $q\TRANSS{\ell} q'$ if $(q,\ell,q')\in \delta$. 
$M$ is {\em connected} if for every state $q\not =q_0$, there is a path 
from $q_0$ to $q$. 
Hereafter we assume each CFSM is connected. 

A CFSM $M=(Q,\chanset,q_0,\ASigma,\delta)$ is {\em deterministic} if for all 
states $q\in Q$ and all actions $\action$, $(q,\action,q'),
(q,\action,q'')\in \delta$ imply $q'=q''$.\footnote{``Deterministic''  
often means the same channel 
should carry a unique value, 
i.e.~if $(q,c!a,q')\in \delta$ 
and $(q,c!a',q'')\in \delta$  
then $a=a'$ and $q'=q''$.  
Here we follow a different definition~\cite{CeceF05} in order to represent branching type constructs.}


\begin{definition}[CS]\label{def:cs}\rm 
A (communicating) system $S$ is a tuple 
$S=(M_\p)_{\p\in\PSet}$ of CFSMs such that 
$M_\p=(Q_\p,\chanset, q_{0\p},\ASigma,\delta_\p)$. 
\end{definition}
%
For $M_\p=(Q_\p,\chanset, q_{0\p},\ASigma,\delta_\p)$, we define a {\em
  configuration} of $S=(M_\p)_{\p\in\PSet}$ to be a tuple
$s=(\vec{q};\vec{w})$ where $\vec{q}=(q_\p)_{\p\in\PSet}$ with $q_\p\in
Q_\p$ and where $\vec{w}=(w_{\p\q})_{\p\neq\q\in\PSet}$ with $w_{\p\q}\in
\ASigma^\ast$.  The element $\vec{q}$ is called a {\em control state} and
$q\in Q_i$ is the {\em local state} of machine $M_i$.
%
\begin{definition}[reachable state]\label{def:rs}\rm 
Let $S$ be a communicating system. 
A configuration $s'=(\vec{q}';\vec{w}')$ is {\em reachable} from 
another configuration $s=(\vec{q};\vec{w})$ by the {\em firing of the
transition $t$}, written $s \TRANS{} s'$ or $s \TRANSS{t} s'$, if 
there exists $a \in \ASigma$ such that either: 
\begin{enumerate} 
\item $t = (q_\p,\p\q!a,q_\p')\in \delta_{\p}$ and 
(a) $q_{\p'}' = q_{\p'}$ for all ${\p'}\not = \p$; and 
(b) $w_{\p\q}' = w_{\p\q}.a$ and 
$w_{\p'\q'}'=w_{\p'\q'}$ for all ${\p'\q'}\not=\p\q$; or  

\item 
  $t = (q_\q,\p\q?a,q_\q')\in \delta_\q$ and 
(a) $q_{\p'}' = q_{\p'}$ for all ${\p'}\not = \q$; and 
(b) $w_{\p\q} = a.w_{\p\q}' $ and $w_{\p'\q'}'=w_{\p'\q'}$ for all ${\p'\q'}\not=\p\q$. 
\end{enumerate}
\end{definition}
The condition (1-b) puts the content $a$ to a channel $\p\q$, 
while (2-b) gets the content $a$ from a channel $\p\q$.  
The reflexive and transitive closure of $\R$ is $\RR$. 
For a transition $t=(s,\action,s')$, we refer to $\ell$ by $\act(t)$.
We write $s_1 \TRANSS{t_1\cdots t_m} s_{m+1}$ 
for $s_1 \TRANSS{t_1} s_2 \cdots \TRANSS{t_m} s_{m+1}$
and 
use $\varphi$ to denote 
$t_1\cdots t_m$. We extend
$\act$ to these sequences: $\act(t_1\cdots t_n)=\act(t_1)\cdots \act(t_n)$.

The {\em initial configuration} of a system is $s_0=
(\vec{q}_0;\vec{\NUL})$ with $\vec{q}_0=(q_{0\p})_{\p\in\PSet}$. 
A {\em final configuration} of the system is 
$s_f= (\vec{q};\vec{\NUL})$ with all $q_\p\in\vec{q}$ final.  A
configuration $s$ is {\em reachable} if $s_0 \RR s$ and we define the {\em
  reachable set} of $S$ as $\RS(S)=\ASET{ s \ | \ s_0 \RR s}$.  We define
the traces of a system $S$ to be $\TR(S)=\ASET{\act(\varphi) \ | \ \exists
  s\in\RS(S), s_0\TRANSS{\varphi} s}$. 

\label{def:basicpropoerties}
We now define several properties about communicating systems and their
configurations. These properties will be used in \S~\ref{sec:cmsa} to
characterise the systems that correspond to multiparty session types.
Let $S$ be a communicating system, $t$ one of its transitions and
$s=(\vec{q};\vec{w})$ one of its configurations.  The following definitions
of configuration properties follow~\cite[Definition 12]{CeceF05}.

\vspace{-1ex}
\begin{enumerate}

\item $s$ is {\em stable} if 
all its buffers are empty, i.e., $\vec{w}=\vec{\NUL}$. 
 
\item 
$s$ is a {\em deadlock configuration} if $s$ is not final, and $\vec{w}=\vec{\NUL}$
and each $q_\p$ is a receiving state, i.e. all machines are 
blocked, waiting for messages. 

\item 
$s$ is an {\em orphan message configuration} if 
all $q_\p\in\vec{q}$ are final 
but $\vec{w}\not = \emptyset$, 
i.e. there is at least an orphan message in a buffer. 

\item 
$s$ is an {\em unspecified reception configuration} 
if there exists $\q\in\PSet$ such that 
$q_\q$ is a receiving state and 
$(q_\q,\p\q?a,q_\q')\in\delta$ implies that 
$|w_{\p\q}| > 0$ and $w_{\p\q} \not\in a\ASigma^\ast$, 
i.e~$q_\q$ is prevented from receiving any message from buffer $\p\q$.
\end{enumerate}
\vspace{-1ex}


\noindent
A sequence of transitions 
is said to be {\em $k$-bounded} if no channel of
any intermediate configuration $s_i$ contains more than $k$
messages. 
We define the $k$-reachability set of $S$ to be the largest
subset $\RS_k(S)$ of $\RS(S)$ within which each configuration $s$ can be
reached by a $k$-bounded execution from $s_0$. Note that, given a
communicating system $S$, for every integer $k$, the set $\RS_k(S)$ is
finite and computable.
We say that a trace $\varphi$ is 
$n$-bound, written $\MEM(\varphi)=n$, if the number of send actions in
$\varphi$ never exceeds the number of receive actions by $n$. 
We then define the equivalences:
\label{def:traceeq}
(1)
$S\approx S'$ is $\forall \varphi,\ \varphi\in\TR(S) \Leftrightarrow\varphi\in\TR(S')$; and 
(2) 
$S\approx_n S'$ is $\forall \varphi,\ \MEM(\varphi)\leq n
  \Rightarrow (\varphi\in\TR(S) \Leftrightarrow\varphi\in\TR(S'))$.
The following key properties will be examined throughout the paper as
properties that multiparty session type can enforce. They are undecidable in
general CFSMs.

\begin{definition}[safety and liveness]\rm 
\label{def:safetyliveness}
(1) 
A communicating system $S$ is {\em deadlock-free} 
(resp.~{\em orphan message-free}, {\em 
reception error-free}) if for all
$s\in \RS(S)$, $s$ is not a deadlock   
(resp.~orphan message, unspecified reception) configuration. 
%
%
%
(2)
$S$ satisfies the {\em liveness property}\footnote{The
    terminology follows \cite{CastagnaDP11}.} if for all $s\in \RS(S)$, 
there exists $s\RED^\ast s'$ 
such that $s'$ is final. 
\end{definition}

\section{Global and local  types: the LTSs and translations}
\label{sec:globallocal}
This section presents the multiparty session types, our main
object of study. For the syntax of types, we follow
\cite{BettiniCDLDY08} which is the most widely used syntax in the 
literature.  We introduce two labelled transition systems, 
for local types and for global types, and show the equivalence
between local types and communicating automata.
%
\ \\[1mm]
{\bf Syntax \ } 
A \emph{global type}, written $G,G',..$, 
describes the whole 
conversation scenario of a multiparty session as a type signature, and
a {\em local type}, written by  $T,T',..$, 
type-abstract sessions from each end-point's view. 
$\p, \q, \dots \in \PSet$ denote participants (see \S~\ref{sec:cfsm}
for conventions). 
The syntax of types is
given as: \vspace{-1ex}
\[
\begin{array}{rclrcl}
G  & \grmeq &
\GBRA{p}{p'}{a_j}{U_j}{G_j}{j\in J} 
\ \grmor\  \mu\typevar.G \ \grmor\ 
            \typevar 
           \ \grmor\ \End\\
T  &\grmeq &
\LRECVK{\pp}{a_i}{U_i}{T_i}{i\in I}
\ \grmor\  \LSENDK{\pp}{a_i}{U_i}{T_i}{i\in I}
\ \grmor\  \mu\typevar.T \ \grmor\ 
            \typevar\ \grmor\ \End 
\end{array}
\]\vspace{-1ex}

\noindent
$a_j\in \ASigma$ corresponds to the usual message label in session type
theory. We omit the mention of the carried types from the syntax in this
paper, as we are not directly concerned by typing processes. 
Global branching type
$\GBRA{p}{p'}{a_j}{U_j}{G_j}{j\in J}$ 
states that participant $\p$ can
send a message with one of the $a_i$ labels to participant $\pp'$ and that interactions
described in $G_j$ follow. 
We require $\p\neq\p'$ to prevent self-sent messages. Recursive type 
$\mu \typevar.G$ is for recursive protocols,
assuming that type variables ($\typevar, \typevar', \dots$) are guarded in the
standard way, i.e. they only occur under branchings. 
Type $\End$ represents session termination (often omitted). 
$\pp\in G$ means that $\pp$ appears in $G$.

Concerning local types, the {\em branching type}
$\LRECVK{\pp}{a_i}{}{T_i}{i\in I}$ specifies the reception of a message from
$\pp$ with a label among the $a_i$.  The {\em selection type}
$\LSENDK{\pp}{a_i}{U_i}{T_i}{i\in I}$ is its dual.  The remaining type
constructors are the same as global types.
When branching is a singleton, we write 
$\TO{p}{p'}:a.{G'}$ for global, and $\pp!a.T$ or $\pp?a.T$ for local. 
%
\ \\[1mm]
{\bf Projection \ } 
The relation between global and local types is formalised by
projection. Instead of the restricted original
projection~\cite{BettiniCDLDY08}, we use the extension with the merging
operator $\mergeop$ from \cite{DY11}: it allows each branch of the
global type to actually contain different interaction patterns.

\begin{definition}[projection]\rm 
\label{def:projection}
The 
{\em projection of $G$ onto
  $\participant{p}$} (written $G\SETRES \participant{p}$)
is defined as:\\[1mm]
$
\begin{array}{c}
\GBRA{p}{p'}{a_j}{U_j}{G_j}{j\in J}\proj{\qq}  = 
\begin{cases}
\LSENDK{\pp}{a_j}{U_j}{\G_j\proj{\qq}}{j\in J} & \qq=\pp \\
\LRECVK{\pp}{a_j}{U_j}{\G_j\proj{\qq}}{j\in J} & \qq=\pp' \\
\sqcup_{j\in J}  \G_j\proj{\qq} & \text{otherwise}
\end{cases}\\
(\mu{\typevar}.{\G})\proj{\pp} 
=
\begin{cases}
\mu{\typevar}.{\G\proj{\pp}} & \G\proj{\pp}\not = \typevar\\
\End & \text{otherwise}
\end{cases} \qquad
\typevar \proj{\pp} \ =  \ \typevar \quad 
\quad \End \proj{\pp} \ =  \ \End
\end{array}
$\\[1mm]
{\em The mergeability relation} $\mergeop$ is the smallest congruence relation over 
local types such that:

\vspace{-1.5em}
$$
\begin{prooftree}
{\forall i\in (K \cap J). T_i\mergeop T_i' \quad 
\forall k\in (K \setminus J),\forall j\in (J \setminus K).
a_k \not = a_j
}
\justifies 
{\p?\{a_k.T_k\}_{k\in K}\mergeop 
\p?\{a_j.T_j'\}_{j\in J}
}
\end{prooftree}
$$
When $T_1\mergeop T_2$ holds, 
we define the operation $\mergecup$ as a partial commutative 
operator over two  types such that $T\mergecup T=T$ for all types and that:
\[
\begin{array}{lll}
\p? \{a_k.T_k\}_{k\in K}
\mergecup 
\p?\{a_j.T_j'\}_{j\in J} \ = 
\ \p ? (\{a_k.(T_k\mergecup T_k')\}_{k\in K\cap J}
\cup \{a_k.T_k\}_{k\in K\setminus J}
\cup \{a_j.T_j'\}_{j\in J\setminus K})\\[1mm]
\end{array}
\]
and homomorphic for other types (i.e. 
$\mathcal{C}[T_1] \sqcup \mathcal{C}[T_2]=\mathcal{C}[T_1\sqcup
T_2]$ where $\mathcal{C}$ is a context for local types). 
We say that $G$ is {\em well-formed} if for all $\p\in \PSet$, 
$G\proj\p$ is defined. 
\end{definition}
\begin{example}[Commit]\label{ex:commit} \rm
The global type for the commit protocol in Figure~\ref{fig:commit} is:
\[
\begin{array}{@{}l@{}l}
\mu \typevar.\TO{\Alice}{\Bob}:& \{\mathit{act}.\,\TO{\Bob}{\Carol}:
\{\mathit{sig}.\,\TO{\Alice}{\Carol}:\mathit{commit}.\typevar \ \}, 
\ \mathit{quit}.\TO{\Bob}{\Carol}:\{\mathit{save}.\TO{\Alice}{\Carol}:\mathit{finish}.\End\}\}
\end{array}
\]
Then $\Carol$'s local type is: 
$\mu \typevar.\Bob? \{\mathit{sig}.\Alice? \{\mathit{commit}.\typevar\},\
\mathit{save}.\Alice? \{\mathit{finish}.\End\} \}$. 
\end{example}
%
\label{sub:LTS}
{\bf LTS over global types}
We next present new labelled transition relations (LTS) for 
global and local types and their sound and complete correspondence.  
The first step for giving a LTS semantics to global types (and then to
local types) is to designate the observables ($\ell, \ell', ...$).  We
choose here to follow the definition of actions for CFSMs where 
a label $\ell$ denotes the sending or the reception of a message of label
$a$ from $\p$ to $\p'$: $ 
\ \ell\grmeq \participant{p}\participant{p'}!a \ | \
\participant{p}\participant{p'}?a
$ 

In order to define an LTS for global types, we need to represent
intermediate states in the execution. For this reason, we introduce in the
grammar of $G$ the construct $\TOS{p}{p'}\colon a_j.G_j$ to represent the
fact that the message $a_j$ has been sent but not yet received.

\begin{definition}[LTS over global types]\rm 
\label{def:glts}
The relation $G\TRANS{\ell} G'$ is defined as ($\sub(\ell)$ is defined in
\S~\ref{def:sub}):\\[1mm]
{\small
$\begin{array}{c}
\rulename{GR1} \quad \GBRA{p}{p'}{a_i}{U_i}{G_i}{i\in I}\
 \TRANS{\p\p'!a_j}\ \GBRAS{p}{p'}{j}{a_i}{U_i}{G_i}{i\in I}
 \quad (j\in I)
\\[1mm]
\rulename{GR2} \quad \GBRAS{p}{p'}{j}{a_i}{U_i}{G_i}{i\in I}
 \TRANS{\p\p'?a_j}\ G_j
\quad\quad
\rulename{GR3} \ 
\infer 
 {
G\subs{\mu \typevar.G}{\typevar}
\TRANS{\ell}G'
}
 {\mu \typevar.G \TRANS\ell G'}\\[1mm]
\rulename{GR4} 
\infer
 {
   \forall j\in I \quad G_j\TRANS{\ell}G_j'\quad \p,\q\not\in \sub(\ell)
}
{\GBRA{p}{q}{a_i}{U_i}{G_i}{i\in I} \TRANS\ell 
\GBRA{p}{q}{a_i}{U_i'}{G_i'}{i\in I}}
\rulename{GR5} 
\infer 
 {
   G_j\TRANS{\ell}G_j' \quad \q\not\in \sub(\ell)\quad \forall i\in
   I\setminus j,G_i' = G_i
}
{\GBRAS{p}{q}{j}{a_i}{U_i}{G_i}{i\in I}
\TRANS\ell 
\GBRAS{p}{q}{j}{a_i}{U_i}{G_i'}{i\in I}
}
  \end{array}$
  }
\end{definition}
\rulename{GR1} represents the emission of a message while
\rulename{GR2} describes the reception of a message.  
\rulename{GR3} governs recursive types. 
\rulename{GR4,5} define the asynchronous semantics of global types,  
where the syntactic order of 
messages is enforced only for the participants that are involved.
For example, in the case when the participants of two consecutive
communications are disjoint, as in: $G_1=\TO{A}{B}:a.  {\TO{C}{D}:b.{\End}}$,
we can observe the emission (and possibly the reception) of $b$ before the
emission (or reception) of $a$ (by \rulename{GR4}).

A more interesting example is: 
$G_2=\TO{A}{B}:a.{\TO{A}{C}:b.\End}$. We write
$\ell_1={A}{B}!a$, 
$\ell_2={A}{B}?a$, 
$\ell_3={A}{C}!b$ and
$\ell_4={A}{C}?b$. 
The LTS allows the following three sequences: 
\[
\footnotesize
\begin{array}{llllllllll}
G_1\ & \!\!\!\!\TRANS{\ell_1}& \TOS{A}{B}:a.\TO{A}{C}:b.{\End}
 & \TRANS{\ell_2} & \TO{A}{C}:b.\End
  & \TRANS{\ell_3} & \TOS{A}{C}:b.\End & \TRANS{\ell_4} \ \End\\[1mm]
G_1\ & \!\!\!\!\TRANS{\ell_1}&  
\TOS{A}{B}:a.\TO{A}{C}:b.{\End}
 & \TRANS{\ell_3} & \TOS{A}{B}:a.\TOS{A}{C}:b.\End & \TRANS{\ell_2} &
\TOS{A}{C}:b.\End & \TRANS{\ell_4}\ \End\\[1mm]
G_1\ & \!\!\!\!\TRANS{\ell_1}& \ 
\TOS{A}{B}:a.\TO{A}{C}:b.{\End}
 & \TRANS{\ell_3} & \TOS{A}{B}:a.\TOS{A}{C}:b.\End & \TRANS{\ell_4} &
\TOS{A}{B}:a.\End & \TRANS{\ell_2}\ \End
\end{array}
\] 
The last sequence is the most interesting: the sender $A$ has to follow the
syntactic order but the receiver $C$ can get the message $b$ before $B$
receives $a$. The respect of these constraints is enforced by the
conditions $\p,\q\not\in \sub(\ell)$ and $\q\not\in \sub(\ell)$ in rules
\rulename{GR4,5}.
\ \\[1mm]
{\bf LTS over local types }
We define the LTS over local types. This is done in two steps,
following the model of CFSMs, where the semantics is given first for
individual automata and then extended to communicating systems.
We use the same labels ($\ell, \ell',...$) as the ones for CFSMs.

\begin{definition}[LTS over local types]
\label{def:llts}
\rm
The relation $T\TRANS{\ell} T'$, for the local type of role $\p$, is defined as:

\vspace{-4ex}\[
\small
\begin{array}{c@{\quad}c@{\quad}c}
\rulename{LR1}\ \LSENDK{\q}{a_i}{U_i}{T_i}{i\in I}
\TRANS{\p\q!a_i} T_i
&
\rulename{LR2}\ 
\LRECVK{\q}{a_i}{U_i}{T_i}{i\in I}
\TRANS{\q\p?a_j}
T_j 
&
\rulename{LR3}\ \infer
 {
T\subs{\mu \typevar.T}{\typevar}
\TRANS{\ell}T'
}
    {\mu \typevar.T \TRANS\ell T'}
\end{array}
\]\vspace{-4ex}
\end{definition}
The semantics of a local type follows the intuition that every action of
the local type should obey the syntactic order.
%
We define the LTS for collections of local types.
\begin{definition}[LTS over collections of local types]
\label{def:cllts}
\rm A configuration $s=(\vec{T};\vec{w})$ of a system of local types
$\ASET{T_\p}_{\p\in \PSet}$ is a pair with $\vec{T}=(T_\p)_{\p\in\PSet}$
and $\vec{w}=(w_{\p\q})_{\p\neq\q\in\PSet}$ with $w_{\p\q} \in
\ASigma^\ast$.  We then define the transition system for configurations. 
For a configuration $s_T=(\vec{T};\vec{w})$, the visible transitions of $s_T \TRANS{\ell}
s_T'=(\vec{T}';\vec{w}')$ are defined as:
\begin{enumerate} 
\item $T_\p 
\TRANS{\p\q!a} 
\T_\p'$ and 
(a) $T_{\p'}' = T_{\p'}$ for all ${\p'}\not = \p$; and 
(b) $w_{\p\q}' = w_{\p\q}\cdot a$ and 
$w_{\p'\q'}'=w_{\p'\q'}$ for all ${\p'\q'}\not=\p\q$; or  

\item 
$T_\q 
\TRANS{\p\q?a} 
\T_\q'$
and 
(a) $T_{\p'}' = T_{\p'}$ for all ${\p'}\not = \q$; and 
(b) $w_{\p\q} = a\cdot w_{\p\q}' $ and $w_{\p'\q'}'=w_{\p'\q'}$ for all ${\p'\q'}\not=\p\q$. 
\end{enumerate}
\end{definition}
The semantics of local types is therefore defined over configurations,
following the definition of the semantics of CFSMs. $w_{\p\q}$
represents the FIFO queue at channel $\p\q$.  
%
We write $\TR(G)$ to denote the set of the visible
traces that can be obtained by reducing $G$. Similarly for $\TR(T)$ and
$\TR(S)$. We extend the trace equivalences $\approx$ and $\approx_n$ 
in \S~\ref{def:traceeq} to
global types and configurations of local types. 
%
%
\label{subsec:lts}

We now state the soundness and completeness of projection with
respect to the LTSs defined above. The proof is given in 
Appendix \ref{app:lts}. 

\begin{theorem}[soundness and completeness]\footnote{
The local type abstracts the behaviour of multiparty typed processes
as proved in the subject reduction theorem in \cite{CHY07}.
Hence this theorem implies that processes typed by global type $G$ 
 by the typing system in \cite{CHY07,BettiniCDLDY08} 
follow the LTS of $G$.}
\label{thm:lts} 
Let $G$ be a global type with participants $\PSet$ and let
$\vec{T}=\ASET{G\proj\p}_{\p\in\PSet}$ be the local types projected from
$G$. Then $G\approx (\vec{T};\vec{\NUL})$.
\end{theorem}
%
%
{\bf Local types and CFSMs \ } 
Next we show 
how to algorithmically go from local types to CFSMs and back while
preserving the trace semantics.
We start by translating local types into CFSMs. 
\begin{definition}[translation from local types to CFSMs]
\rm 
Write $T'\in T$ if $T'$ occurs in $T$. 
Let $T_0$ be the local type of participant $\pp$
projected from $G$. The automaton corresponding to $T_0$ is
$\A(T_0)=(Q,C,q_0,\ASigma,\delta)$ where:
(1) $Q = \{ T'\ | \ T'\in T_0, \  T' \not = \typevar, T' \not = \mu \typevar.T  \}$;
(2) $q_0= T_0'$ with $T_0 = \mu \vec{\typevar}.T_0'$ and $T_0'\in Q$;   
(3) $C = \{ \p\q \PAR \p,\q \in \G\}$; 
(4) $\ASigma$ is the set of $\{a \in G\}$; and  
(5) $\delta$ is defined as:\\[1mm]
\begin{tabular}{l}
If $T = \pp'! \{a_j.T_j\}_{j\in J}\in Q$,
then $
\begin{cases}
(T,(\p\p'! a_j),T_{j})\in \delta & T_j \not = \typevar\\
(T,(\p\p'! a_j), T')\in\delta & 
T_j = \typevar, \ \ \mu \typevar\vec{\typevar}.T'\in T_0,  T'\in Q 
\end{cases}
$
\\
If $T = \pp' ?\{ a_j.T_j\}_{j\in J}\in Q$,  
then $
\begin{cases}
(T,(\p'\p? a_j),T_{j})\in \delta & T_j \not = \typevar\\
(T,(\p'\p? a_j),T')\in\delta & 
T_j = \typevar, \ \mu \typevar\vec{\typevar}.T'\in T_0,  T'\in Q 
\end{cases}
$
\end{tabular}
\end{definition}
The definition says that 
the set of states $Q$ 
are the suboccurrences of branching or selection or $\End$ in 
the local type; the initial state $q_0$ is 
the occurrence of (the recursion body of) $T_0$;  
the channels and alphabets 
correspond to those in $T_0$; and the transition is defined from the
state $T$ to its body $T_j$ with the action $\p\p'!a_j$ for the 
output and $\p\p'?a_j$ for the input. If $T_j$ is a recursive 
type variable $\typevar$, it points the state of the body of 
the corresponding recursive type. As an example of the translation, see $\mathtt{C}$'s local type in 
Example \ref{ex:commit} and its corresponding automaton in Figure~\ref{fig:commit}.


\begin{proposition}[local types to CFSMs]
\label{pro:basic}
Assume $T_\p$ is a local type. Then $\A(T_\p)$ is deterministic, 
directed and has no mixed states. 
\end{proposition}



\noindent We say that a CFSM is {\em basic} if it is 
deterministic,  directed and has no mixed states. 
Any basic CFSM can be translated into a local type.  

\begin{definition}[translation from a basic CFSM to a local type]\rm
Let $M_\p=(Q,C,q_0,\ASigma,\delta)$ and assume 
$M_\q$ is basic. Then we define the translation $\LT(M_\p)$ such that 
$\LT(M_\p)=\LT_\NUL(q_0)$ where $\LT_{\VEC{q}}(q)$ is defined as:
%
%
%
%
%
%
%
%
%
\vspace{-2mm}
\begin{itemize}
\item[(1)]
$\LT_{\VEC{q}}(q)=\mu\typevar_q.\p'!\{
  a_j.\LT_{\VEC{q}\cdot q}^\circ(q_{j})\}_{j\in J}$
if $(q,\p\p'!a_j,q_{j}) \in \delta$; 
\item[(2)] 
$\LT_{\VEC{q}}(q)=
  \mu \typevar_q.\p'? \{
a_j.\LT_{\VEC{q}\cdot q}^\circ(q_{j})\}_{j\in J}$
if $(q,\p'\p?a_j,q_{j}) \in \delta$;
\item[(3)]
$\LT^\circ_{\VEC{q}}(q)=\LT_\NUL(q)=\End$ if $q$ is final;  
(4)
$\LT^\circ_{\VEC{q}}(q)=\typevar_{q_k}$ if $(q,\ell,q_k) \in
  \delta$ and $q_k\in \VEC{q}$; 
and 
\item[(5)] 
$\LT_{\VEC{q}}^\circ(q)=\LT_{\VEC{q}}(q)$ otherwise. 
\end{itemize}
Finally, we replace $\mu \typevar.T$ by $T$ if $\typevar$ is not in $T$.
\end{definition} 
In $\LT_{\VEC{q}}$, $\VEC{q}$ records visited states; 
(1,2) translate the receiving and sending states to branching and selection types, respectively; (3) translates the final state to $\End$; and 
(4) is the case of a recursion: since $q_k$ was visited, 
$\ell$ is dropped and replaced by the type variable. 

The following states that the translations preserve the semantics.

\begin{proposition}[translations between CFSMs and local types]
\label{pro:translation}
If a CFSM $M$ is basic, then $M\WB\LT(M)$. If $T$ is a local
type, then $T\WB\A(T)$.
\end{proposition}



\section{Completeness and synthesis}
\label{sec:cmsa}
This section studies the synthesis and sound and complete characterisation
of the multiparty session types as communicating automata. We first note that
basic CFSMs correspond to the natural generalisation of half-duplex
systems~\cite[\S~4.1.1]{CeceF05}, in which each pair of machines linked by
two channels, one in each direction, communicates in a half-duplex way. In
this class, the safety properties of Definition~\ref{def:safetyliveness}
are however undecidable~\cite[Theorem 36]{CeceF05}. We therefore need 
a stronger (and decidable) property to force basic CFSMs to behave as
if they were the result of a projection from global types.
\\[1mm]
{\bf Multiparty compatibility}
In the two machines case, there exists a sound and complete condition
called {\em compatible}~\cite{GMY84}. Let us define the isomorphism
$\Phi:(\chanset \times \ASET{!,?} \times {\Bbb A})^\ast \RED (\chanset \times
\ASET{!,?} \times {\Bbb A})^\ast$ such that $\Phi(j?a)=j!a$,
$\Phi(j!a)=j?a$,
$\Phi(\NUL)=\NUL$, 
$\Phi(t_1\cdots t_n)=\Phi(t_1)\cdots \Phi(t_n)$.
 $\Phi$ exchanges a sending action with the corresponding
receiving one and vice versa.  The compatibility of two machines can be
immediately defined as $\TR(M_1)=\Phi(\TR(M_2))$ (i.e. the traces of $M_1$
are exactly the set of dual traces of $M_2$).  
The idea of the extension to the multiparty case comes from the observation
that from the viewpoint of the participant $\p$, the rest of all the
machines $(M_\q)_{\q\in \PSet\setminus \p}$ should behave as if they
were one CFSM which offers compatible traces $\Phi(\TR(M_\p))$, up to
internal synchronisations (i.e. 1-bounded executions). Below we define a way to group CFSMs.

\begin{definition}[Definition 37, \cite{CeceF05}]\rm 
Let $M_i=(Q_i,C_i,q_{0i},{\Bbb A}_i,\delta_i)$. The {\em associated CFSM} of
  $S=(M_1,..,M_n)$ is $M=(Q,C,q_0,\Sigma,\delta)$ such that:   
$Q=Q_1\times Q_2 \times \cdots \times  Q_n$, $q_0=(q_{01},\dots,q_{0n})$ 
and $\delta$ is the least relation verifying: 
$((q_1,...,q_i,...,q_n),\action,(q_1,...,q_i',...,q_n))\in \delta$ if 
$(q_i,\action,q_i')\in \delta_i$ ($1\leq i\leq n$). 
\end{definition}

Below we define a notion of compatibility extended to more than two
CFSMs.
We say that $\varphi$ is an {\em alternation} 
if $\varphi$ is an alternation of sending and 
corresponding receive actions (i.e. the action
$\p\q!a$ is immediately followed by $\p\q?a$). 


\begin{definition}[multiparty compatible system]\rm 
\label{def:multipartycompatible}
A system $S=(M_1,..,M_n)$ ($n\geq 2$) is
{\em multiparty compatible} if for any 1-bounded reachable stable state $s\in
\RS_1(S)$, 
for any sequence of actions $\ell_1\cdots \ell_k$
from $s$ in $M_i$, there is a sequence of transitions 
$\varphi_1\cdot t_1 \cdot \varphi_2\cdot t_2
\cdot \varphi_3
\cdots \varphi_k \cdot t_k$ from $s$ in 
a CFSM corresponding to $S^{-i}=(M_1,..,M_{i-1},M_{i+1},..,M_n)$ where
$\varphi_j$ is either empty or an alternation,  $\ell_j = \Phi(\act(t_j))$
and $i\not\in\act(\varphi_j)$ for $1\leq  j \leq k$ 
(i.e.~$\varphi_j$ does not contain actions to or from channel $i$).

\end{definition}
The above definition states that for each $M_i$, 
the rest of machines $S^{-i}$ can produce the compatible (dual) actions 
by executing alternations in $S^{-i}$. 
From $M_i$, these intermediate alternations can be seen as non-observable
internal actions. 


\begin{example}[multiparty compatibility]\rm
\label{ex:mcomp}
As an example, we can test the multiparty compatibility property on the commit example
of Figure~\ref{fig:commit}.  
We only detail here how to check the
compatibility from the point of view of $\Carol$. To check the compatibility for
the actions $\act(t_1\cdot t_2)=\mathit{BC}?\text{sig}\cdot\mathit{AC}!\text{commit}$, the only possible 1-bound (i.e.~alternating) execution is
$\mathit{AB}!\text{act}\cdot\mathit{AB}?\text{act}$, and
$\Phi(\act(t_1))=\mathit{BC}!\text{sig}$ sent from $\Bob$
and $\Phi(\act(t_2))=\mathit{AC}!\text{commit}$ sent from $\Alice$. 
To check the compatibility for
the actions $\act(t_3\cdot t_4)=\mathit{BC}?\text{save}\cdot\mathit{AC}?\text{finish}$,
the 1-bound execution is
$\mathit{AB}!\text{quit}\cdot\mathit{AB}?\text{quit}$, and
$\Phi(\act(t_3))=\mathit{BC}!\text{save}$ from $\Bob$
and $\Phi(\act(t_4))=\mathit{AC}!\text{finish}$ from $\Alice$. 
\end{example}
\begin{remark}
\label{rem:non-unique-sender}
In Definition \ref{def:multipartycompatible}, 
we require to check the compatibility from any 1-bounded reachable
stable state 
in the case one branch is selected by different senders. 
Consider the following machines:\\
$
\begin{array}{l|l} 
\mathtt{A}\rightarrow\hspace{-1cm}
\xymatrix@C=2em@R=2ex{
   & *++[o][F]{}\ar[r]^{BA?\text{a}}\ar[d]_{BA?\text{b}}
   & *++[o][F]{}\ar[r]^{CA?\text{c}}
   & *++[o][F=]{}  \\
   & *++[o][F]{}\ar[r]^{CA?\text{d}} & *++[o][F=]{} \\
 } 
\quad
\mathtt{B}\rightarrow\hspace{-1cm}
\xymatrix@C=2em@R=2ex{
   & *++[o][F]{}\ar[r]^{BA!\text{a}}\ar[d]_{BA!\text{b}}
   & *++[o][F=]{}\\
   & *++[o][F=]{} \\
 } 
\quad
\mathtt{C}\rightarrow\hspace{-1cm}
\xymatrix@C=2em@R=2ex{
   & *++[o][F]{}\ar[r]^{CA!\text{c}}\ar[d]_{CA!\text{d}}
   & *++[o][F=]{}\\
   & *++[o][F=]{} \\
 } \ & \
\mathtt{A'}\rightarrow\hspace{-1cm}
\xymatrix@C=2em@R=2ex{
   & *++[o][F]{}\ar@/^/[r]^{BA?\text{a}}@\ar@/_/[r]_{BA?\text{b}}
   & *++[o][F]{}\ar[r]^{CA?\text{c}}\ar[dr]_{CA?\text{d}}
   & *++[o][F=]{}  \\
   &  & & *++[o][F=]{} \\
 } 
\end{array}
$\\[2mm]
In  $\Alice$, $\Bob$ and $\Carol$, 
each action in each machine has its dual 
but they do not satisfy multiparty compatibility. For example, 
if $BA!\text{a}\cdot BA?\text{a}$ is executed, 
$CA!\text{d}$ does not have a dual action (hence 
they do not satisfy the safety properties). 
On the other hand,  the machines $\Alice'$, $\Bob$ and $\Carol$ 
satisfy the multiparty compatibility. 
\end{remark}

\begin{theorem} 
\label{pro:core:stable}
\label{pro:cmsa:safety}
Assume $S=(M_\p)_{\p \in \PSet}$ is basic and multiparty compatible. 
Then 
$S$ satisfies the three safety properties 
in Definition
\ref{def:safetyliveness}.
Further, if there exists at least one $M_\q$ which
includes a final state, then $S$ satisfies the liveness property. 
\end{theorem}
\begin{proof}
\ifpopllong
We first prove that any basic $S$ which satisfies multiparty compatible is
{\em stable} ($S$ is stable, if, for all $s\in \RS(S)$, 
there exists an execution $\TRANSS{\varphi'}$ such that 
$s\TRANSS{\varphi'} s'$ and $s'$ is stable, and 
there is a 1-bounded execution $s_0\TRANSS{\varphi''} s'$, i.e.
any trace can be translated into a 1-bounded execution
after some appropriate executions). The proof is non-trivial
using a detailed analysis of causal relations   
to translate into a 1-bounded executions. 
Then the orphan message- and the reception error-freedom are its
corollary. 
The deadlock-freedom is proved by the stable property and 
multiparty compatibility. Liveness 
is a consequence of the orphan message- and 
deadlock-freedom. See Appendix \ref{app:cmsa}. 
\QED
\else 
See \cite{fullversion}. 
\fi 
\end{proof}

\begin{proposition}
If all the CFSMs $M_\p$ ($\p\in \PSet$) are basic, 
there is an algorithm to check whether $(M_\p)_{\p\in \PSet}$
is multiparty compatible. 
\end{proposition}
\begin{proof}
The algorithm to check $M_\p$'s compatibility with $S^{-\p}$ is
defined using the set $\RS_1(S)$ of reachable states using 1-bounded
executions. 
Note that the set $\RS_1(S)$ 
is decidable \cite[Remark 19]{CeceF05}.  
We start from $q=q_0$ and the initial configuration $s=s_0$. 
Suppose that, from $q$, we have the transitions $t_i=(q,\q\p!a_i,q'_i) \in
\delta_\p$.
We then construct $\RS_1(S)$ (without executing $\p$) until it includes
$s'$ such that $\ASET{s'\TRANSS{t_i}\TRANSS{t'_j} s_j}_{j\in J}$ where
$\act(t'_i)=\q\p?a_i$ and $I\subseteq J$.  If there exists no such $s'$, it
returns false and terminates.  The case where, from $q$, we have receiving
transitions $t=(q,\q\p?a_i,q'_i)$ is dual.
If it does not fail, we continue to check from state $q'_i$ and
configuration $s_i$ for each $i\in I$. 
We repeat this procedure until we visit all $q\in Q_\p$. Then repeat for
the other machines $\p'$ such that $\p'\in \PSet\setminus \p$. Then we
repeat this procedure for all stable $s\in \RS_1(S)$. 
\QED 
\end{proof}
%
{\bf Synthesis \ }
Below we state the lemma which will be crucial for the proof of
the synthesis and completeness. 
The lemma comes from the intuition that the transitions of multiparty
compatible systems are always permutations of one-bounded executions 
as it is the case in multiparty session types. 
See Appendix \ref{app:one_to_k} for the proof.  


\begin{lemma}[1-buffer equivalence]\label{lem:one_to_k}
Suppose $S_1$ and $S_2$ are two basic and multiparty compatible
communicating systems such that $S_1\approx_1 S_2$, then $S_1\approx S_2$.
\end{lemma}




\begin{theorem}[synthesis] \label{thm:classical:syn}
Suppose $S$ is a basic system and multiparty compatible. 
Then there is an algorithm which successfully builds well-formed $G$ such that 
$S\approx G$ if such $G$ exists, and otherwise terminates. 
\end{theorem}
\vspace*{-5mm}
\begin{proof}
We assume $S=(M_\p)_{\p\in \PSet}$.
The algorithm starts from the initial states of all machines
  $(q^{\p_1}{}_{0},...,q^{\p_n}{}_{0})$.
We take a pair of the initial states which 
is a sending state $q^\p_{0}$ and a receiving state $q^\q_{0}$
from $\p$ to $\q$. 
We note that by directness, if there are more than two pairs, 
the participants in two pairs are disjoint, and by [G4] in Definition 
\ref{def:glts}, the order does not matter. 
We apply the algorithm with the 
invariant that all buffers are empty and that we repeatedly pick up one
pair such that $q_{\p}$ (sending state) and $q_{\q}$ (receiving state). 
We define $G(q_{1},...,q_{n})$ where ($q_{\p},q_{\q}\in
\{q_{1},...,q_{n}\}$) 
as follows:
  \begin{itemize}
    \item if $(q_{1},...,q_{n})$ has already been examined and if all
      participants have been involved since then (or the ones that have not
      are in their final state), we set $G(q_{1},...,q_{n})$ to be
      $\typevar_{q_{1},...,q_{n}}$.
      Otherwise, we select a pair sender/receiver from two participants
      that have not been involved (and are not final) and go to the next step;
    \item otherwise, in $q_{\p}$, from machine $\p$, we know that all the
      transitions are sending actions towards $\p'$ (by directedness),
      i.e. of the form $(q_{\p}, \p\q!a_i, q_i)\in \delta_\p$ for
      $i\in I$. 
      \begin{itemize}
      \item we check that machine $\q$ is in a receiving state $q_\q$ such
        that $(q_\q,\p\q?a_j,q'_j)\in \delta_{\p'}$ with $j\in J$ and
        $I\subseteq J$.  
        
      \item we set $\mu \typevar_{q_{1},...,q_{n}}.\TO{p}{q}\colon
        \ASET{a_i.G(q_{1},...,q_\p\isnow q_i,...,q_\q\isnow q_i',...,
            q_{n})}_{i\in I}$ (we replace $q_\p$ and $q_\q$ by $q_i$ and
          $q_i'$, respectively) and continue by recursive calls. 

      \item if all sending states in $q_{1},...,q_{n}$ become final, then we set 
        $G(q_{1},...,q_{n})=\End$.
      \end{itemize}
    \item we erase unnecessary $\mu \typevar$ if $\typevar\not\in G$ 
and check $G$ satisfies Definition \ref{def:projection}.
    \end{itemize}
    Since the algorithm only explores 1-bounded executions, 
    the reconstructed $G$ satisfies $G\approx_1 S$. By
    Theorem~\ref{thm:lts}, we know
    that $G\approx (\ASET{G\proj\p}_{\p\in\PSet};\vec{\NUL})$. Hence, by
    Proposition~\ref{pro:translation}, we have 
    $G\approx S'$ where $S'$ is the communicating system translated
    from the projected local types $\ASET{G\proj\p}_{\p\in\PSet}$ of $G$. By
    Lemma~\ref{lem:one_to_k}, $S\approx S'$ and therefore $S\WB
    G$. \QED
\end{proof}

\noindent The algorithm can generate the global type in Example~\ref{ex:commit} 
from CFSMs in Figure~\ref{fig:commit} and the global type 
$\Bob\to\Alice\{a:\Carol\to\Alice:\{c:\End,d:\End\}, b:\Carol\to\Alice:\{c:\End,d:\End\}\}$ 
from $\Alice'$, $\Bob$ and $\Carol$ in Remark \ref{rem:non-unique-sender}. 
Note that 
$\Bob\to\Alice\{a:\Carol\to\Alice:\{c:\End \}, b:\Carol\to\Alice:\{d:\End\}\}$ 
generated by 
$\Alice$, $\Bob$ and $\Carol$ in Remark \ref{rem:non-unique-sender}
is not projectable by Definition~\ref{def:projection}, hence it is not well-formed. 
%

By Theorems \ref{thm:lts},  
\ref{pro:cmsa:safety} and  
\ref{thm:classical:syn},  
and Proposition \ref{pro:translation}, 
we can now conclude:

\begin{theorem}[soundness and completeness in \MSACore]
\label{thm:soundcompletemsacore}
Suppose $S$ is 
basic and multiparty compatible. 
Then there exists $G$ such that 
$S \WB G$.
Conversely, if $G$ is well-formed,  
then there exists $S$ which 
satisfies the three safety properties in Definition
\ref{def:safetyliveness}
and $S\WB G$.  
\end{theorem}
\section{Conclusion and related work} 
\label{sec:related}
\label{sec:conclusion}
This paper investigated the sound and complete characterisation of 
multiparty session types
into
CFSMs and developed a decidable synthesis algorithm from basic CFSMs. 
The main tool we used is a new extension to multiparty interactions of the
duality condition for binary session types, called {\em multiparty
  compatibility}. The basic condition (coming from the 
binary session types) and the multiparty compatibility property are a {\em
  necessary and sufficient condition} to obtain safe global types.
Our aim is to offer a duality notion 
which would be applicable to extend 
other theoretical foundations 
such as the Curry-Howard correspondence with 
linear logics \cite{CF10,Wadler2012} to multiparty
communications. Basic multiparty compatible CFSMs also define
one of the few non-trivial decidable subclass of CFSMs which satisfy deadlock-freedom.
The methods proposed here 
are palatable to a wide range of applications 
based on choreography protocol models 
and more widely, finite state machines. 
We are currently working on two applications based on the theory
developed in this paper: the Testable Architecture 
\cite{savara} which enables the communication structure of the
implementation to be inferred and to be tested against the
choreography; and dynamic monitoring for a large scale
cyberinfrastructure in \cite{ooi} where a central controller can check
that distributed update paths for monitor specifications (which form 
FSMs projected from a global specification) are safe by synthesis.

Our previous work~\cite{DY12} 
presented the first translation from global and local types into CFSMs.
It only analysed the 
properties of the automata resulting from such a translation.
The complete characterisation of global types independently from the
projected local types was left open, as was synthesis. 
This present paper closes this open
problem. 
There are a large number of paper that can be found in the literature about
the synthesis of CFSMs. See~\cite{AM10} for a summary of recent results. 
The main distinction with CFSM synthesis is, apart from 
the formal setting (i.e.~types), 
about the kind of the target specifications to be generated 
(global types in our case). Not only our synthesis
is concerned about trace properties (languages) like the standard synthesis
of CFSMs (the problem of the closed synthesis of CFSMs is usually defined as the
construction from a regular language $L$ of a machine satisfying certain
conditions related to buffer boundedness, deadlock-freedom and words swapping),  but we also generate concrete syntax or choreography descriptions 
as {\em types} of programs or software. 
Hence they are directly applicable to programming languages 
and can be straightforwardly 
integrated into the existing frameworks that are based on session types. 

Within the context of multiparty session types, \cite{esop09} first
studied the reconstruction of a global type from its projected local
types up to asynchronous subtyping and \cite{LT12} recently offers a
typing system to synthesise global types from local types.
Our synthesis based on CFSMs is more general 
since CFSMs do not depend on the syntax. For example, 
\cite{esop09,LT12} cannot treat the synthesis for 
$\Alice'$, $\Bob$ and $\Carol$ in Remark \ref{rem:non-unique-sender}. 
These works also do not study the completeness (i.e.~they build a
global type from a set of projected local types (up to subtyping),
and do not investigate necessary and sufficient conditions to 
build a well-formed global type). A difficulty of
the completeness result is that it is generally unknown if the global
type constructed by the synthesis can simulate executions with arbitrary
buffer bounds since the synthesis only directly looks at 1-bounded executions. 
In this paper, we proved Lemma~\ref{lem:one_to_k} and bridged this gap
towards the complete characterisation. 
Recent work by \cite{CastagnaDP11,Tevfik12} focus on 
proving the semantic correspondence between global and local
descriptions (see \cite{DY12} for more detailed comparison), but  
no synthesis algorithm 
is studied.

\vspace*{-2mm}

{\footnotesize
\bibliographystyle{abbrv}
\bibliography{session}
}

\appendix 

\ifpopllong
\section{Appendix for Section \ref{sec:globallocal}}

\subsection{Proof of Theorem \ref{thm:lts}}
\label{app:lts}

\myparagraph{Local Types Subtyping}

In order to relate global and local types, we define in
Figure~\ref{fig:subtyping} a subtyping relation $\prec$ on local types. 
Local type $T'$ is a super type of local type $T$, written $T\prec
T'$, if it offers more receive transitions. We note that $T_i \prec \sqcup_{i\in I}
T_i$.

\begin{figure} \large
\[
\begin{array}{l@{\quad}l@{\quad}l@{\quad}l}
\frac{\forall i\in I, T_i \prec T_i'}{\LSENDK{\pp}{a_i}{U_i}{T_i}{i\in I} \prec
  \LSENDK{\pp}{a_i}{U_i}{T_i'}{i\in I}} &
\frac{I\subseteq J \quad \forall i\in I, T_i \prec
  T_i'}{\LRECVK{\pp}{a_i}{U_i}{T_i}{i\in I} \prec
  \LRECVK{\pp}{a_j}{U_j}{T_j'}{j\in J}} &
\frac{}{\typevar\prec\typevar} &
\frac{T\prec T'}{\mu\typevar.T\prec\mu\typevar.T'}
\end{array}
\]
\caption{Subtyping between local types}\label{fig:subtyping}
\end{figure}

This subtyping relation can be extended to configurations in the following way:
$(\vec{T};\vec{w})\prec(\vec{T'};\vec{w'})$ if $\vec{w}=\vec{w'}$ and $\forall
\p\in\PSet, T_\p \prec T_\p'$.

The main properties of subtyping is that it preserves traces, i.e. if $s\prec
s'$, then $s \approx s'$.

\myparagraph{Extension of projection}

In order to prove Theorem~\ref{thm:lts}, we extend the definition of
projection to global intermediate states.

We represent the projected configuration $\projrel{G}$ of a global type $G$ as a
configuration
$\ASET{G\proj\p}_{\p\in\PSet},\projrel{G}_{\ASET{\NUL}_{\q\q'\in\PSet}}$
where 
the content of the buffers $\projrel{G}_{\ASET{\NUL}_{\q\q'\in\PSet}}$ is
given by:
\[
\begin{array}{rcl}
  \projrel{\TOS{\p}{\p'}\colon a_j.G_j }_{\ASET{w_{\q\q'}}_{\q\q'\in\PSet}} & = &
  \projrel{G_j}_{\ASET{w_{\q\q'}}_{\q\q'\in\PSet}[w_{\p\p'}=w_{\p\p'}\cdot
    a_j]}\\
  \projrel{\TO{\p}{\p'}\colon a_j.G_j}_{\ASET{w_{\q\q'}}_{\q\q'\in\PSet}} & = &
  \projrel{G_j}_{\ASET{w_{\q\q'}}_{\q\q'\in\PSet}}\\
  \projrel{\GBRA{p}{p'}{a_j}{U_j}{G_j}{j\in J}}_{\ASET{w_{\q\q'}}_{\q\q'\in\PSet}} & = &
  \projrel{G_1}_{\ASET{w_{\q\q'}}_{\q\q'\in\PSet}}
  \\
  \projrel{\mu\typevar.G}_{\ASET{w_{\q\q'}}_{\q\q'\in\PSet}} & = & \ASET{w_{\q\q'}}_{\q\q'\in\PSet}
  \\
  \projrel{\End}_{\ASET{w_{\q\q'}}_{\q\q'\in\PSet}} & = & \ASET{w_{\q\q'}}_{\q\q'\in\PSet}
\end{array}
\]
and where the projection algorithm $\proj \q$ is extended by:
\[
\begin{array}{c}
\GBRAS{p}{p'}{j}{a_i}{U_i}{G_i}{i\in I}
\proj{\qq}  = 
\begin{cases}
\LRECVK{\pp}{a_i}{U_i}{\G_i\proj{\qq}}{i\in I} & \qq=\pp' \\
\G_j\proj{\qq} & \text{otherwise}\\
\end{cases}
\end{array}
\]

This extended projection allows us to match global type and projected
local type transitions step by step.

\myparagraph{Theorem~\ref{thm:lts}}

We prove Theorem~\ref{thm:lts} by combining the local type subtyping and
extended projection into a step equivalence lemma.
Theorem~\ref{thm:lts} is a simple consequence of Lemma~\ref{lem:stepequivalence}.

\begin{lemma}[Step equivalence]\label{lem:stepequivalence}
  For all global type $G$ and local configuration $s$, if $\projrel{G}\prec s$,
  then we have $G\TRANSS{\ell}G' \Leftrightarrow s\TRANSS{\ell}s'$ and $\projrel{G'}\prec s'$.
\end{lemma}
\begin{proof}
The proof is by induction on the possible global and local transitions.

\paragraph{\bf Correctness} By induction on the structure of each reduction
  $G\TRANS{\ell}G'$, we prove that $\projrel{G}\TRANS{\ell}s$ with
  $\projrel{G'}\prec s$. We use the fact that if $s\prec s'$, then $s \approx
  s'$, to consider only matching transition for $\projrel{G}$.
  \begin{itemize}
  \item[\rulename{GR1}] where $G=\GBRA{p}{p'}{a_i}{U_i}{G_i}{i\in I}\
    \TRANSS{\p\p'!a_j}\ G'=\GBRAS{p}{p'}{j}{a_i}{U_i}{G_i}{i\in I}$.
    The projection of $G$ is $\projrel{G}=s_T=\ASET{T_\q}_{\q\in
      \PSet},\ASET{w_{\q\q'}}_{\q\q'\in\PSet}$.  The local types are: $T_\p=
    G\proj\p= \LSENDK{\p'}{a_i}{U_i}{\G_i\proj{\p}}{i\in I}$ and $T_{\p'}=
    G\proj\p'=\LRECVK{\p}{a_i}{U_i}{\G_i\proj{\p'}}{i\in I}$ and (for
    $\q\notin\{\p,\p'\}$) $T_\q=\sqcup_{i\in I} \G_i\proj{\qq}$.  Rule
    \rulename{LR1} allows $\LSENDK{\p'}{a_i}{}{\G_i\proj{\p}}{i\in
      I}\TRANS{\p\p'!a_j}\G_j\proj{\p}$.  We therefore have
    $s_T\TRANS{\p\p'!a_j}\ASET{T'_\q}_{\q\in
      \PSet},\ASET{w'_{\q\q'}}_{\q\q'\in\PSet}$, with $T'_\q=T_\q$ if
    $\q\neq\p$, and $T'_\p=\G_j\proj{\p}$, and with $w'_{\q\q'}=w_{\q\q'}$ if
    $\q\q'\neq\p\p'$, and $w'_{\p\p'}=w_{\p\p'}\cdot a_j$. 

    Since $\G_j\proj{\qq} \prec \sqcup_{i\in I} \G_i\proj{\qq}$, we have
    $\ASET{T'_\q}_{\q\in \PSet},\ASET{w'_{\q\q'}}_{\q\q'\in\PSet} \prec
    \projrel{G}$. 

    This corresponds exactly to the projection $\projrel{G'}$ of $G'$.
  \item[\rulename{GR2}] where $G=\GBRAS{p}{p'}{j}{a_i}{U_i}{G_i}{i\in I}$
    $\TRANSS{\p\p'?a_j}\ G'=G_j$. The projection of $G$ is
    $\projrel{G}=s_T=\ASET{T_\q}_{\q\in
      \PSet},\ASET{w_{\q\q'}}_{\q\q'\in\PSet}$.  The local types are:
    $T_\p= G\proj\p= \G_j\proj{\p}$ and $T_{\p'}=
    G\proj\p'=\LRECVK{\p}{a_j}{U_j}{\G_j\proj{\p'}}{}$ and (for
    $\q\notin\{\p,\p'\}$) $T_\q=\G_j\proj{\qq}$. We also know that
    $w_{\p\p'}$ is of the form $w'_{\p\p'}\cdot a_j$.

    Using \rulename{LR2}, $\ASET{T_\q}_{\q\in
      \PSet},\ASET{w_{\q\q'}}_{\q\q'\in\PSet}\TRANSS{\p\p'?a_j}\
    \ASET{G_j\proj\q}_{\q\in \PSet},\ASET{w'_{\q\q'}}_{\q\q'\in\PSet}$ with
    $w'_{\q\q'}=w_{\q\q'}$ if $\q\q'\neq\p\p'$.
    The result of the transition is the same as the projection
    $\projrel{G'}$ of $G'$.
  \item[\rulename{GR3}] where $G=\mu\typevar.G' \TRANSS{\ell} G''$. 

    By hypothesis, we know that $G'\subs{\typevar}{\mu\typevar.G'}
    \TRANSS{\ell} G''$. By induction, we know that
    $\projrel{G'\subs{\typevar}{\mu\typevar.G'}}=s_T=\ASET{T_\q}_{\q\in
      \PSet},\ASET{w_{\q\q'}}_{\q\q'\in\PSet}$ can do a reduction
    $\TRANSS{\ell}$ to $\projrel{G''}=s_T=\ASET{T'_\q}_{\q\in
      \PSet},\ASET{w'_{\q\q'}}_{\q\q'\in\PSet}$. Projection is homomorphic
    for recursion, hence $G'\subs{\mu\typevar.G'}{\typevar}\proj\q =
    G'\proj\q\subs{\mu\typevar.G'\proj\q}{\typevar}$.  We use
    \rulename{LR4} to conclude.
  \item[\rulename{GR4}] where $\GBRA{p}{q}{a_i}{U_i}{G_i}{i\in I}
    \TRANS\ell \GBRA{p}{q}{a_i}{U_i}{G_i'}{i\in I}$ and
    $\pp,\qq\notin\subj(\ell)$. By induction, we know that, $\forall i\in
    I, \projrel{G_i} \TRANSS{\ell} \projrel{G_i'}$. We need to prove that
    $\projrel{\GBRA{p}{q}{a_i}{U_i}{G_i}{i\in
        I}}\TRANS\ell\projrel{\GBRA{p}{q}{a_i}{U_i}{G_i'}{i\in I}}$.  The
    projections for all participants are identical, except for
    $\qq'=\subj(\ell)$, whose projection is (computed by merging)
    $\sqcup_{i\in I} \G_i\proj{\qq'}$. Since $\forall i\in I, \projrel{G_i}
    \TRANSS{\ell} \projrel{G_i'}$, we know that all the $G_i\proj{\qq'}$
    have at least the prefix corresponding to $\ell$, and that, using
    either $\rulename{LR1}$ or $\rulename{LR2}$, the continuations are the
    $G_i'\proj{\qq'}$. We can then conclude that the $\sqcup_{i\in I}
    \G_i\proj{\qq'}\TRANSS{\ell} \sqcup_{i\in I} \G_i'\proj{\qq'}$.
  \item[\rulename{GR5}] where $\GBRAS{p}{q}{j}{a_i}{U_i}{G_i}{i\in I} \TRANS\ell
    \GBRAS{p}{q}{j}{a_i}{U_i}{G_i'}{i\in I}$ and $\qq\notin\subj(\ell)$ with
    $G'_i=G_i$ for $i\neq j$. By induction, we know that, $\projrel{G_j}
    \TRANSS{\ell} \projrel{G_j'}$. We need to prove that
    $\projrel{\GBRAS{p}{q}{j}{a_i}{U_i}{G_i}{i\in
        I}}\TRANS\ell\projrel{\GBRA{p}{q}{j}{a_i}{U_i}{G_i'}{i\in I}}$.  The
    projections for all participants are identical, except for
    $\qq'=\subj(\ell)$, whose projection is
    $\G_j\proj{\qq'}$. By induction,
    $\G_j\proj{\qq'}\TRANS\ell\G_j'\proj{\qq'}$, which allows us to conclude.
  \end{itemize}

\paragraph{\bf Completeness} \ We prove by induction on
  $\projrel{G}=$\\$\ASET{T_\p}_{\p\in
    \PSet},\ASET{w_{\q\q'}}_{\q\q'\in\PSet}\TRANS{\ell}\ASET{T'_\p}_{\p\in
    \PSet},\ASET{w'_{\q\q'}}_{\q\q'\in\PSet}$ that $G\TRANS{\ell}G'$ with
  $\projrel{G'}\prec\ASET{T'_\p}_{\p\in
    \PSet},\ASET{w'_{\q\q'}}_{\q\q'\in\PSet}$.

  \begin{itemize}
  \item[\rulename{LR1}] There is $T_\p= G\proj\p=
    \LSENDK{\p'}{a_i}{U_i}{\G_i\proj{\p}}{i\in I}$. By definition of
    projection, $G$ has $\GBRA{p}{q}{a_i}{U_i}{G_i}{i\in I}$ as subterm,
    possibly several times (by mergeability). By definition
    of projection, we note that no action in $G$ can involve $\p$ before any of the
    occurrences of $\GBRA{p}{q}{a_i}{U_i}{G_i}{i\in I}$. Therefore we can
    apply as many times as needed \rulename{GR4} and \rulename{GR5}, and use \rulename{GR1} to reduce to
    $\TOS{p}{q}\colon a_j.G_j$. The projection of the
    resulting global type corresponds to a subtype to the result of \rulename{LR1}.
  \item[\rulename{LR2}] There is $T_\p=
    G\proj\p=\LRECVK{\qq}{a_j}{U_j}{\G_j\proj{\p}}{j\in J}$.  To activate
    \rulename{LR2}, there should be a value $a_j$ in the buffer
    $w_{\p\q}$. By definition of projection, $G$ has therefore
    $\GBRAS{p}{q}{j}{a_i}{U_i}{G_i}{i\in I}$ as subterm, possibly several
    times (by mergeability). By definition of projection, no action in $G$
    can involve $\p$ before any of the occurrences of
    $\GBRAS{p}{q}{j}{a_i}{U_i}{G_i}{i\in I}$. We can apply as many times as
    needed \rulename{GR4} and \rulename{GR5} and use \rulename{GR2} to
    reduce to $G_j$. The projection of the resulting global type
    corresponds to the result of \rulename{LR2}.

  \item[\rulename{LR3}] where $T=\mu\typevar.T'$. Projection is homomorphic with
    respect to recursion. Therefore $G$ is of the same form. We can use
    \rulename{GR3} and induction to conclude.
  \end{itemize}
\end{proof}


\subsection{Local types and CFSMs}

\myparagraph{Proposition \ref{pro:basic}} 
For the determinism, we note that all $a_i$ in 
$\LRECVK{\pp}{a_i}{}{T_i}{i\in I}$ 
and $\LSENDK{\pp}{a_i}{U_i}{T_i}{i\in I}$ are distinct. 
Directdness is by the syntax of branching and selection types. 
Finally, for non-mixed states, we can check a state is either  
sending or receiving state as one state represents either 
branching and selection type. 

\myparagraph{Proposition \ref{pro:translation}} 
The first clause is by the induction of $M$ using the translation of $\LT$. The second clause is by the induction of $T$ using the translation of $\A$. Both are mechanical.

\section{Appendix for Section \ref{sec:cmsa}}
\label{app:cmsa}
We say that a configuration $s$ with $t_1$ and $t_2$ satisfies the {\em
  one-step diamond property} if, assuming $s \TRANSS{t_1} s_1$ and $s
\TRANSS{t_2} s_2$ with $t_1\not = t_2$, there exists $s'$ such that $s_1
\TRANSS{t_1'} s'$ and $s_2 \TRANSS{t_2'} s'$ where $\act(t_1)=\act(t_2')$
and $\act(t_2)=\act(t_1')$. We use the following lemma to permute the two 
actions. 

\begin{lemma}[diamond property in basic machines]
\label{lem:MSAdiamond}
Suppose $S=(M_\p)_{\p \in \PSet}$ and $S$ is basic. 
Assume $s\in \RS(S)$ and $s   \TRANSS{t_1} s_1$ and $s \TRANSS{t_2} s_2$. 
\begin{enumerate}
\item If $t_1$ and $t_2$ are both sending actions such that 
$\act(t_1)= \p_1\q_1! a_1$ and $\act(t_2)= \p_2\q_2! a_2$, we have
either:  
\begin{enumerate}
\item $\p_1=\p_2$ and $\q_1=\q_2$ and $a_1=a_2$ with $s_1=s_2$; 
\item $\p_1=\p_2$ and $\q_1=\q_2$ and $a_1\not =a_2$;
\item $\p_1\not=\p_2$ and $\q_1\not =\q_2$ with $a_1 \not=a_2$, and 
$s$ with $t_1$ and $t_2$ satisfies the diamond property. 
\end{enumerate}
\item  If $t_1$ and $t_2$ are both receiving actions such that 
$\act(t_1)= \p_1\q_1? a_1$ and $\act(t_2)= \p_2\q_2? a_2$, we have
either:  
\begin{enumerate}
\item $\p_1=\p_2$ and $\q_1=\q_2$ and $a_1=a_2$ with $s_1=s_2$; 
\item $\p_1\not=\p_2$ and $\q_1\not =\q_2$ with $s_1 \not=s_2$, and 
$s$ with $t_1$ and $t_2$ satisfies the diamond property. 
\end{enumerate}
\item  If $t_1$ is a receiving action and $t_2$ is a sending action such that 
$\act(t_1)= \p_1\q_1? a_1$ and $\act(t_2)= \p_2\q_2! a_2$, we have
either: 
\begin{enumerate}
\item $\q_1=\q_2$ and $\p_1 \not =\p_2$; or 
\item $\p_1=\p_2$ and $\q_1 \not =\q_2$; or 
\item $\p_1\not=\p_2$ and $\q_1\not =\q_2$ 
\end{enumerate}
with $s_1 \not=s_2$, and 
$s$ with $t_1$ and $t_2$ satisfies the diamond property. 
\end{enumerate}
\vspace{-2ex}
\end{lemma}
\begin{proof}
For (1), 
there is no 
case such that $\p_1 \neq \p_2$ and  $\q_1 = \q_2$ since $S$ is directed.  
Then if 
$\p_1=\p_2$ and $\q_1=\q_2$ and $a_1=a_2$, then $s_1=s_2$ by the determinism.  
For (2), 
there is no 
case such that $\p_1 \neq \p_2$ and  $\q_1 = \q_2$ since $S$ is directed. 
Also there is no case such that 
$\p_1=\p_2$ and $\q_1=\q_2$ and $a_1\not =a_2$ since 
the communication between the same peer is done via an FIFO queue. 
For (3), there is no case such that 
$\q_1=\q_2$ and $\p_1 =\p_2$ because of no-mixed state. \QED 
\end{proof}

The following definition aims to explicitly describe 
the causality relation between the actions. These are useful 
to identify the permutable actions. 

\begin{definition}[causality]\rm
\label{def:causality}
\begin{enumerate}
\item Suppose $s_0 \TRANSS{\varphi} s$ and $\varphi=\varphi_0 \cdot t_1 \cdot 
\varphi_1 \cdot t_2 \cdot \varphi_2$. 
We write $t_1 \triangleleft t_2$ ($t_2$ depends on $t_1$) if 
either (1) $t_1 = \p\q!a$ and $t_1 = \p\q?a$ for some $\p$ and $\q$ 
or (2) $\subj(t_1)=\subj(t_2)$. 

\item We say $\varphi=t_0\cdot t_1 \cdot t_2 \cdots t_n$ is {\em the causal 
chain} if $s_0 \TRANSS{\varphi'} s'$ and $\varphi\subseteq \varphi'$ 
with, for all $0 \leq k \leq n-1$, there exists $i$ such that $i> k$ 
and $t_k \triangleleft t_i$. We call $\varphi$ the maximum causal
chain if there is no causal chain $\varphi''$ such that $\varphi\subsetneq 
\varphi''\subseteq \varphi'$. 

\item Suppose $s_0 \TRANSS{\varphi} s$ and $\varphi=\varphi_0 \cdot t_1 \cdot 
\varphi_1 \cdot t_2 \cdot \varphi_2$. We write $t_i\sharp t_j$ if 
there is no causal chain from $t_i$ to $t_j$ with $i < j$. 
\end{enumerate}
\end{definition}

By Lemma \ref{lem:MSAdiamond}, we have: 

\begin{lemma}[maximum causality]
\label{lem:causality}
Suppose $S$ is basic and $s\in \RS(S)$. Then for all
$s\TRANSS{\varphi}s'$, we have $s\TRANSS{\varphi_m\cdot \varphi''}s'$
and $s\TRANSS{\varphi''\cdot \varphi_m'}s'$ where
$\varphi_m,\varphi_m'$ are the maximum causal chain. 
\end{lemma}

\begin{lemma}[output-input dependency]
\label{lem:OI}
Suppose $S$ is basic. Then 
there is no causal chain $t_0\cdot t_1 \cdot t_2 \cdots t_n$ 
such that $\act(t_0)=\p\q!a$ and 
$\act(t_n)=\p\q'?b$ with $a\not = b$ and $\act(t_i)\not = \p\q?c$ for any
$c$ 
($1\leq i
\leq n-1$). 
\end{lemma}
\begin{proof}
We use the following definition. The causal chain 
$\varphi
=t_0\cdot t_1 \cdots t_n$ is called 
\begin{enumerate}
\item {\em O-causal chain} if for all  $1\leq i\leq n$, 
$t_i=\p\q_i!a_i$ with some $\q_i$ and $a_i$. 
\item {\em I-causal chain} if for all  $1\leq i\leq n$, 
$t_i=\q_i\p?a_i$ with some $\q_i$ and $a_i$. 
\end{enumerate}
Then any single causal chain 
$\varphi=\VEC{t}_0\cdot \VEC{t}_1 \cdots \VEC{t}_n$ 
can be decomposed into 
alternating O and I causal chains where 
$t_i = \cdot t_{i0} \cdots t_{in_i}$
 with either 
(1) $\act(t_{in_i})=\p\q!a$ and $\act(t_{i+10})=\q'\p?b$; 
(2) $\act(t_{in_i})=\p\q?a$ and $\act(t_{i+10})=\q\p'!b$; or 
(3) $\act(t_{in_i})=\p\q!a$ and $\act(t_{i+10})=\p\q?a$. 
In the case of (1,2), we note $\subj(t_{ih})=\subj(t_{i+1k})$
for all  $0\leq h \leq n_i$ and  $0\leq k \leq n_{i+1}$. 

Now assume $S$ is basic 
and there is a sequence $\varphi
=t_0\cdot t_1 \cdots t_n$ such that 
$\act(t_0)=\p_0\q_0!a_0$ and 
$\act(t_n)=\p_n\q_n?a_n$ with 
$\p_0 = \q_n$,  
$a_0\not = a_n$ 
and $\act(t_i)\not = \p_0\q_0?a$ for any $a$ ($1\leq i
\leq n-1$). We prove $\varphi$ is not a causal chain by the induction of
the length of $\varphi$. \\[1mm]
{\bf Case} $n=1$. By definition, $t_0\sharp t_n$.\\[1mm] 
{\bf Case} $n>1$. If $\varphi$ is a causal chain, 
there is a decomposition into O and I causal chains such that 
$\varphi=
\VEC{t}_0\cdot \VEC{t}_1 \cdots \VEC{t}_m$
where 
$t_i = t_{i0} \cdots t_{in_i}$. 
By the condition 
$t_i\not = \p_0\q_0?a$ for any $a$ ($1\leq i
\leq n-1$), the case (3) above is excluded. 
Hence 
we have 
$\subj(t_{ih})=\subj(t_{i+1k})$
for all  $0\leq h \leq n_i$ and  $0\leq k \leq n_{i+1}$. 
This implies 
\begin{enumerate}
\item 
$\p_0=\p_{ij}$ with $i$ even (in the O causal chains)
\item 
$\q_{ij}=\q_0$ with $i$ odd (in the I causal chains); and 
\item 
$\p_{in_i}=\q_{i+10}$ with $i$ even. 
\end{enumerate}
This implies $\p_0=\q_0$ which contradicts 
the definition of the channels of CFSMs (i.e. 
$\p_0\not=\q_0$ if $\p_0\q_0$ is a channel). 
Hence there is no causal chain from 
$\act(t_0)=\p_0\q_0!a_0$ to 
$\act(t_n)=\p_0\q_0?a_n$ 
if $\act(t_i)\not = \p_0\q_0?a$ and $a_0\not = a_n$. 
\end{proof}

\begin{lemma}[input availablity]
\label{lem:cmsa:inputavailable}
Assume $S=(M_\p)_{\p \in \PSet}$ is basic and multiparty compatible. 
Then for all 
$s\in \RS(S)$, if
$s\TRANSS{\p\p'!a}s'$, then 
$s'\TRANSS{\varphi}s_2\TRANSS{\p\p'?a} s_3$. 
\end{lemma}
\begin{proof}
We use Lemma \ref{lem:MSAdiamond} and Lemma \ref{lem:causality}. 
Suppose $s\in \RS(S)$ and $s\TRANSS{t}s'$ such that 
$\act(t)=\p\p'!a$. By contradiction, assume there is no $\varphi'$ 
such that $s'\TRANSS{\varphi'}\TRANSS{t'}s''$ with $\act(t)=\p\p'?a$. 
Then there should be some input state $(q,\q\p'?b,q')\in \delta_{\p'}$
where $q\TRANSS{\q\p'?b'}q''\TRANSS{p_1}\TRANSS{\p\p'?a}q'''$ where
$b\not = b'$ (hence $q'\not = q''$ by determinism), i.e. $\q\p'?b$
leads to an incompatible path with one lead to the action $\q\p'?a$.

Suppose $s'\TRANSS{\varphi_0}\TRANSS{t_{bi}}s''$  
with $t_{bi}=(q,\q\p'?b,q')$. Then 
$\varphi_0$ should include 
the corresponding output action $\act(t_{bo})=\q\p'!b$. 
By Lemma \ref{lem:causality}, without loss of generality, we assume 
$\varphi_0\cdot t_{bi}$ is the maximum causal chain to 
$t_{bi}$. 
Let us write $\varphi_0 = t_0 \triangleleft t_1 \triangleleft \cdots
\triangleleft t_{n}$. By Lemma \ref{lem:MSAdiamond}, 
we can set $t_{bo}=t_{n}$. 
Note that for all $i$, $\act(t_i)\not = \p\p'?a'$ by the assumption:
since if 
$\act(t_i)\not = \p\p'?a$, then it contradicts the assumption such that 
$t$ does not have a corresponding input; and 
if $\act(t_i) = \p\p'?a'$ with $a\not= a'$ then, 
by directedness of $S$, 
it contradicts 
to the assumption that $t_{bi}$ is the first input which leads to the
incompatible path. 
Then there are three cases. 

\begin{enumerate}
\item there is a chain from $t$ to $t_{n}=t_{bo}$, 
i.e. there exists 
$0\leq i\leq n$ such that $t \triangleleft t_i\triangleleft \cdots
\triangleleft t_{n}$. 

\item there is no direct chain from $t$ to $t_{n}$ but there is 
a chain to $t_{bi}$, 
i.e. there exists $0\leq i\leq n$ 
such that $t \triangleleft t_i \triangleleft \cdots
\triangleleft t_{bi}$. 

\item there is no chain from $t$ to either $t_{n}$ or $t_{bi}$. 
\end{enumerate}
{\bf Case 1:} By the assumption, there is no $t_j$ such that
$\act(t_j)=\p\p'?a'$. Hence $t_i=\p\p''!a'$ for some $a'$ and $\p''$.\\[1mm] 
{\bf Case 1-1:} there is no input in $t_j$ in 
$t \triangleleft t_i\triangleleft \cdots
\triangleleft t_{n-1}$. Then $\p=\q$, i.e.~$\q\p'!b=\p\p'!b$. Then by
the definition of $s\TRANS{t} s'$ (i.e. by FIFO semantics at each
channel),  
$\p\p'?b$ cannot perform 
before $\p\p'?a$. This case contradicts to the assumption $\p\p'?a$ is
not available. \\[1mm]
{\bf Case 1-2:} there is an input $t_j$ in 
$t \triangleleft t_i\triangleleft \cdots
\triangleleft t_{n-1}$. By $t \triangleleft t_i$, 
$\subj(\act(t_i))=\p$. Hence  
we have either $\act(t_i)=\p\q_i!a_i$ with $\q\not = \q_i$ or 
$\act(t_i)=\q_i\p?a_i$. \\[1mm]
{\bf Case 1-2-1:} $\act(t_i)=\p\q_i!a_i$. Then there is 
a path $q\TRANS{\p\q!a}\TRANS{\p\q_i!a_i}q'$ in $M_\p$. Hence by the
multiparty compatibility, there should be the traces 
$\p\q?a\cdot \varphi \cdot \p\q_i?a_i$ with $\varphi$ alternation from 
the machine with respect to $\{ M_\participant{r}\}_{
\participant{r} \in \PSet\setminus \p}$. 
This contradicts to the assumption that $\p\p'?a$ is
not available. 
\\[1mm]
{\bf Case 1-2-2:} 
$\act(t_i)=\q_i\p?a_i$. 
Similarly with the case {\bf Case 1-2-1},
by the multiparty compatibility, there should be the traces 
$\p\q?a\cdot \varphi \cdot \p\q_i?a_i$ with $\varphi$ alternation from 
the machine with respect to $\{ M_\participant{r}\}_{
\participant{r} \in \PSet\setminus \p}$. Hence it contradicts to the
assumption.  
\\[1mm]
{\bf Case 2:} 
Assume the chain 
such that $t \triangleleft t_i \triangleleft \cdots
\triangleleft t_{bi}$ and $t\sharp t_n$. 
As the same reasoning as {\bf Case 1}, $\p \not = \q$ 
and $t_i$ is either $\p\q_i!a_i$ or $\q_i\p?a_i$.  
Then we use the multiparty compatibility.  
\\[1mm]
{\bf Case 3:} 
Suppose 
there exists $s_{04}\in \RS(S)$ such that
$s_{04}\TRANSS{t_4}\TRANSS{\varphi_{4}}\TRANSS{\varphi_{0}}\TRANSS{t_{bi}}$ 
and 
$s_{04}\TRANSS{t_4'}\TRANSS{\varphi_{4}'}\TRANSS{t}$ 
where $t_4$ leads to $t_{bi}$ and 
$t_4'$ leads to $t$. \\[1mm]
{\bf Case 3-1:} 
Suppose $t_4$ and $t_4'$ are both sending actions.    
By Lemma \ref{lem:MSAdiamond}, there are three cases. \\[1mm]
{\bf (a)} This case does not satisfy the assumption since $s_1=s_2$.\\[1mm]
{\bf (b)} We set $\act(t_4) = \p_4\q_4!d$ and 
$\act(t_4') = \p_4\q_4!d'$ with $d\not = d'$. In this case, 
we cannot execute both $t$ and $t_{bi}$. Hence there is no possible
way to execute $t_{bi}$. This contradicts to the assumption.\\[1mm]  
{\bf (c)} Since this case satisfy the diamond property, we apply the
same routine from $s'$ such that
$s_{04}\TRANSS{t_4}\TRANSS{t_41}s'$ and 
$s_{04}\TRANSS{t_4'}\TRANSS{t_42}s'$ and 
$\act(t_4)=\TRANSS{t_42}$ and 
$\act(t_4')=\TRANSS{t_41}$ where the length of the sequences to $t$
and $t_{bi}$ is reduced (hence this case is eventually matched with 
other cases). 
\\[1mm] 
{\bf Case 3-2:} 
Suppose $t_4$ and $t_4'$ are both sending actions.    
By Lemma \ref{lem:MSAdiamond}, there are two cases. The case {\bf (a)} is
as the same as the case {\bf 3-1-(b)} and 
the case {\bf (b)} is
as the same as the case {\bf 3-1-(c)}. \\[1mm]
{\bf Case 3-3:} 
Suppose $t_4$ is a sending action and $t_4'$ is receiving action. 
This case is as the same as the case {\bf 3-1-(c)} and  
This  concludes the proof. 
 \QED
\end{proof}

We can extend the above lemma. The proof is similar.  

\begin{lemma}[general input availablity]
\label{lem:cmsa:generalinputavailable}
Assume $S=(M_\p)_{\p \in \PSet}$ is basic and multiparty compatible. 
Then for all 
$s\in \RS(S)$, if
$s\TRANSS{\p\p'!a}s_1\TRANSS{\varphi}s'$ with $\p\p'?a\not\in \varphi$, then 
$s'\TRANSS{\varphi'}s_2\TRANSS{\p\p'?a} s_3$. 
\end{lemma}







\subsection{Proofs of Theorem \ref{pro:cmsa:safety}}
\label{app:cmsa:proof:stable}
We first prove the following stable property. 

\begin{proposition}[stable property]
Assume $S=(M_\p)_{\p \in \PSet}$ is basic and multiparty compatible. 
Then 
$S$ satisfies the stable property, i.e. 
if, for all $s\in \RS(S)$, 
there exists an execution $\TRANSS{\varphi'}$ such that 
$s\TRANSS{\varphi'} s'$ and $s'$ is stable, and 
there is a 1-bounded execution $s_0\TRANSS{\varphi''} s'$.
\end{proposition}
\begin{proof}
We proceed by the induction of the total number of messages (sending actions)
which should be closed by the corresponding received actions. 
Once all messages are closed, we can obtain 1-bound execution. 

Suppose $s_1,s_2$ are the states such that 
$s_0\TRANSS{\varphi_1} s_1 \TRANSS{t_1} s_2 \TRANSS{\varphi_1'} s'$ 
where $\varphi_1$ is a 1-bounded execution and 
$s_1 \TRANSS{t_1} s_2$ is the first transition which is not followed 
by the corresponding received action. 
Since $\varphi_1$ is a 1-bounded execution, 
there is $s_3$ such that $s_2 \TRANSS{t_2} s_3$ where $t_1$ and $t_2$
are both sending actions. 
Then by the definition of the compatibility and Lemma \ref{lem:cmsa:inputavailable}, we have 
\begin{equation}
\label{proof:inputavailability}
s_1 \TRANSS{t_1} s_2 \TRANSS{\varphi_2}\TRANSS{\overline{t_1}}s_3'
\end{equation}
where $\varphi_2$ is an alternation execution and
$\overline{t_1}=\p\q?a$. Assume $\varphi_2$ is a minimum execution
which leads to $\overline{t_1}$. We need to show 
\[
s_1 \TRANSS{\varphi_2}\TRANSS{t_1}\TRANSS{\overline{t_1}}
s_3'\TRANSS{t_2} s_4
\]
Then we can apply the same routine for $t_2$ to close 
it by the corresponding receiving action $\OL{t_2}$. 
Applying this to the next sending state one by one,  
we can reach an 1-bounded execution. 
Let $\varphi_2= t_4\cdot \varphi_2'$. Then 
by the definition  of multiparty compatibility, 
$\act(t_4)=\p'\q'!c$ and $\p'\not =\p$ and $\q'\not =\q$. 
Hence by Lemma \ref{lem:MSAdiamond}(1), there exists the execution 
such that 
\[ 
s_1 \TRANS{t_4}\TRANS{t_1}\TRANS{\varphi_2'}\TRANS{\overline{t_1}}
s_3'\TRANS{t_2} s_4
\]
Let $\varphi_2'={\overline{t_4}}\cdot\varphi_2''$ where 
$\overline{t_1}=\p'\q'?c$. Then this time, by Lemma \ref{lem:MSAdiamond}(2), 
we have: 
\[ 
s_1 \TRANS{t_4}\TRANS{\overline{t_4}}\TRANS{t_1} \TRANS{\varphi_2''}\TRANS{\overline{t_1}}
s_3'\TRANS{t_2} s_4
\]
where $\varphi_1 \cdot t_4 \cdot \overline{t_4}$ is a 1-bounded
execution. Applying this permutation repeatedly, we have 
\[ 
s_1 \TRANS{\varphi_3}\TRANS{t_1}\TRANS{\overline{t_1}}
s_3'\TRANS{t_2} s_4
\]
where $\varphi_3$ is an 1-bounded execution. We apply the same routine 
for $t_2$ and conclude $s_1 \TRANS{\varphi'} s'$ for some stable
$s'$. \QED 
\end{proof}

From the stable property, 
the orphan message- and the reception error-freedom are immediate. 
Also the liveness 
is a corollary by the orphan message- and deadlock-freedom.
Hence we only prove the deadlock-freedom assuming the stable property. 

\myparagraph{Deadlock-freedom}
Assume $S$ is basic and satisfy the multiparty session compatibility. 
By the above lemma, $S$ satisfies the stable property. 
Hence we only have to check 
for all $s\in \RS_1(S)$, 
$s$ is not dead-lock. Suppose by the contradiction, $s$ contains 
the receiving states $t_1,...,t_n$. Then by the multiparty compatibility, 
there exists 1-bounded execution $\varphi$ such that 
$s\TRANS{\varphi}\TRANS{\overline{t}_1}s'$. Hence $s'\TRANS{t_1} s''$
and $s''$ is stable. Applying this routine 
to the rest of receiving states $t_2,...,t_n$, we conclude the proof.
\QED 

\subsection{Proof for Lemma \ref{lem:one_to_k}}
\label{app:one_to_k}

\begin{proof}
  We prove by induction that $\forall n, S_1\approx_n S_2 \implies
  S_1\approx_{n+1} S_2$. Then the lemma follows.  

  We assume $S_1\approx_n S_2$ and then prove, by induction on the length
  of any execution $\varphi$ that uses less than $n$ buffer space in $S_1$,
  that $\varphi$ is accepted by $S_2$. If the length $|\varphi|<n+1$, then the buffer
  usage of $\varphi$ for $S_1$ cannot exceed $n$, therefore $S_2$ can
  realise $\varphi$ since $S_1\approx_n S_2$.

  Assume that a trace $\varphi$ in $S_1$ has length $|\varphi|=k+1$, that
  $\varphi$ is $(n+1)$-bound, and that any trace strictly shorter than $\varphi$
  or using less buffer space is accepted by $S_2$.

  We denote the last action of $\varphi$ as $\ell$. We name $\ell_0$
  the last unmatched send transition $\p\q!a$ of $\varphi$ that is not
  $\ell$. We can therefore write $\varphi$ as $\varphi_0 \ell_0
  \varphi_1 \ell$, with $\varphi_1$ minimal. 
 I.e.~there is no permutation 
such that $\varphi_0 \ell \varphi_0' \ell_0$. 
In $S_1$, we have
\begin{equation}\label{eq:main}
S_1: \ s_0
\TRANS{\varphi_0}\TRANS{\ell_0}\TRANS{\varphi_1}s_1\TRANS{\ell}
s
\end{equation}
By Lemma~\ref{lem:cmsa:generalinputavailable},
we have a trace $\varphi_2$ such that:
\begin{equation}\label{eq:second}
S_1: \ s_0 \TRANS{\varphi_0}\TRANS{\ell_0}\TRANS{\varphi_1}s_1
\TRANS{\varphi_2}\TRANS{\overline{\ell_0}}s_1'
\end{equation}
{\bf Case} $\varphi_2 = \NUL$. Hence 
\begin{equation}\label{eq:NUL}
S_1: \ s_0 \TRANS{\varphi_0}\TRANS{\ell_0}\TRANS{\varphi_1}s_1
\TRANS{\overline{\ell_0}}s_1' \quad 
\text{and} \quad 
s_1 \TRANS{\ell}s
\end{equation}
Let $\ell=\p_1\q_1!b$. Then by 
Lemma \ref{lem:MSAdiamond} (3),  
$s_1 \TRANS{\OL{\ell_0}}\TRANS{\ell}s''$ as required. \\[1mm]
{\bf Case} $\varphi_2 = \ell_1\cdot\varphi_2'$. 
\begin{enumerate}
\item If $\ell=\p_1\q_1!b$ and $\ell_1=\p_2\q_2?c$, then 
by Lemma \ref{lem:MSAdiamond} (3),  
$s_1 \TRANS{\ell_1}\TRANS{\ell}s''$. Hence we apply the induction on
$\varphi_2'$. 

\item If $\ell=\p_1\q_1!b$ and $\ell_1=\p_2\q_2!c$, then by
directedness, we have three cases: 
\begin{enumerate}
\item $\p_1\not = \p_2$ and $\q_1\not = \q_2$. 
By Lemma \ref{lem:MSAdiamond} (1), we have 

\begin{equation}\label{eq:INDone}
s_1\TRANS{\ell_2}s \TRANS{\ell} s_2' \TRANS{\varphi_2'} s_1' 
\end{equation}
Hence we conclude by the induction on $\varphi_2'$. 

\item $\p_1 = \p_2$ and $\q_1 = \q_2$ and $b\not = c$. 

In this case, by Lemma~\ref{lem:cmsa:generalinputavailable},
there exists $\varphi_3$ such that $s_1 \TRANS{\ell}\TRANS{\varphi_3}
\TRANS{\OL{\ell_0}}$. Hence this case is subsumed into (a) or (c) below. 

\item $\p_1 = \p_2$ and $\q_1 = \q_2$ and $b=c$. 

Since $\ell_0$ and $\ell$ is not permutable, there 
is the causality such that  
$t_0 \triangleleft t_1 \triangleleft \cdots
\triangleleft t_n \triangleleft \cdots \triangleleft t_{n+m}$ 
with 
$\act({t_0})=\ell_0$, 
$\act({t_n})=\ell$ and 
$\act({t_{n+m}})=\OL{\ell_0}$. 
We note that since $l_0$ is the first outstanding output, 
by multiparty compatibility, 
$t_i$ ($1\leq i \leq n-1$)  does not include $\p_1\q_1?a$.   
Then by Lemma \ref{lem:OI}, this case does not exist. 



\end{enumerate}
\end{enumerate}
Applying Case (a), we can build in $S_1$ a sequence of
transitions that allows $\ell$ using strictly less buffer space as: 
\begin{equation}
S_1: \ s_0 \TRANS{\varphi_0}
\TRANS{\varphi_0'}\TRANS{\ell_0}\TRANS{\varphi_3}\TRANS{\overline{\ell_0}}\TRANS{\ell}
\end{equation}
where $\varphi_3$ is the result of the combination of $\varphi_1$ and
$\varphi_2$ using commutation.

By the assumption ($S_1\WB_n S_2$), $S_2$ can simulate this sequence as:
\begin{equation}
\label{S2}
S_2: \ s_0 \TRANS{\varphi_0}
\TRANS{\varphi_0'}\TRANS{\ell_0}\TRANS{\varphi_3}\TRANS{\overline{\ell_0}}\TRANS{\ell}
\end{equation}
All the commutation steps used in $S_1$ are also valid in $S_2$ 
since they are solely based on
causalities of the transition sequences.  
We therefore can permute (\ref{S2}) back to:
\begin{equation}
S_2: \ s_0 \TRANS{\varphi_0}\TRANS{\ell_0}\TRANS{\varphi_3}\TRANS{\ell}
\end{equation}
It concludes this proof.

\end{proof}
\section{Generalised Multiparty Session Automata}
\label{sec:graph}

As an addition to the main results, we extend the results obtained on classical
multiparty session types to tackle generalised multiparty session
types~\cite{DY12}, an extension with new features such as flexible fork, choice,
merge and join operations for precise flow specification. It strictly subsumes
classical MPST.

\subsection{Generalised global and local types}
\label{sec:graphs-global}

In this subsection, we recall definitions from~\cite{DY12}.

\myparagraph{Generalised global types} We first define {\em generalised
  global types}. The syntax is defined \iflong in
Fig.~\ref{fig:globaltypes}.  \else below.  \fi \iflong\begin{figure}[t]\fi
\iflong  \mycaption{Generalised Global Types \label{fig:globaltypes}}\fi
{\small
\[
      \begin{array}{ll}
        \begin{array}{r@{}c@{\quad}l@{\quad}l@{}}
        \multicolumn{4}{l}{
        \GG \ \grmeq \ \defk\ \GV \  \In\  \xx \qquad
        \text{Global type}} \\[1mm]
        G & \grmeq & \xx =  \TO{p}{p'}:a\;;\xx' &
        \text{Messages} \\
        & \grmor & \xx = \xx' \PAR \xx'' & \text{Fork}\\
        & \grmor & \xx \PAR \xx' = \xx'' & \text{Join}\\
      \end{array} &
      \begin{array}{r@{}c@{\quad}l@{\quad}l@{}}
        \\[1mm]
        & \grmor & \xx =  \End & \text{End}\\
        & \grmor & \xx = \xx' + \xx'' & \text{Choice}\\
        & \grmor & \xx + \xx' = \xx'' & \text{Merge}\\
      \end{array}

    \end{array}
\]}
%
\iflong\hrule\else
\fi
\iflong \end{figure}
\fi %


A global type $\GG = \defk\ \GV \ \In\ \xx_0$ describes an interaction
between a fixed number of participants. We explain each of the constructs
by example, in Figure~\ref{fig:data-global}, alongside the corresponding
graphical representation inspired by the BPMN 2.0 business processing
language. This example features three participants, with $A$ sending data
to $B$ while $C$ concurrently records a log entry of the transmission.

\begin{figure}[t]
\[
\begin{array}{@{}l@{}l@{}}
  \begin{array}{rcl}
 \GG= \defk\ \xx_0 & = & \xx_1 \PAR \xx_2 \\
   \xx_1 + \xx_5 & = & \xx_3 \\
   \xx_3 & = & \TO{A}{B}:\textit{data}\;;\xx_4  \\
   \xx_4 & = & \xx_5 + \xx_6  \\
   \xx_6 & = & \TO{A}{B}:\textit{eof}\;;\xx_7  \\
   \xx_2 & = & \TO{A}{C}:\textit{log}\;;\xx_8  \\
   \xx_7 \PAR \xx_8 & = & \xx_9 \\
   \xx_9 & = & \TO{B}{C}:\textit{save}\;;\xx_{10}  \\
   \xx_{10} & = & \End \ \In\ \xx_0\\[2ex]
   \multicolumn{3}{c}{\text{Data transfer example}}
  \end{array} &
    \begin{minipage}{15em}
      \includegraphics[trim=6cm 9cm 5cm 11cm,width=12.5em]{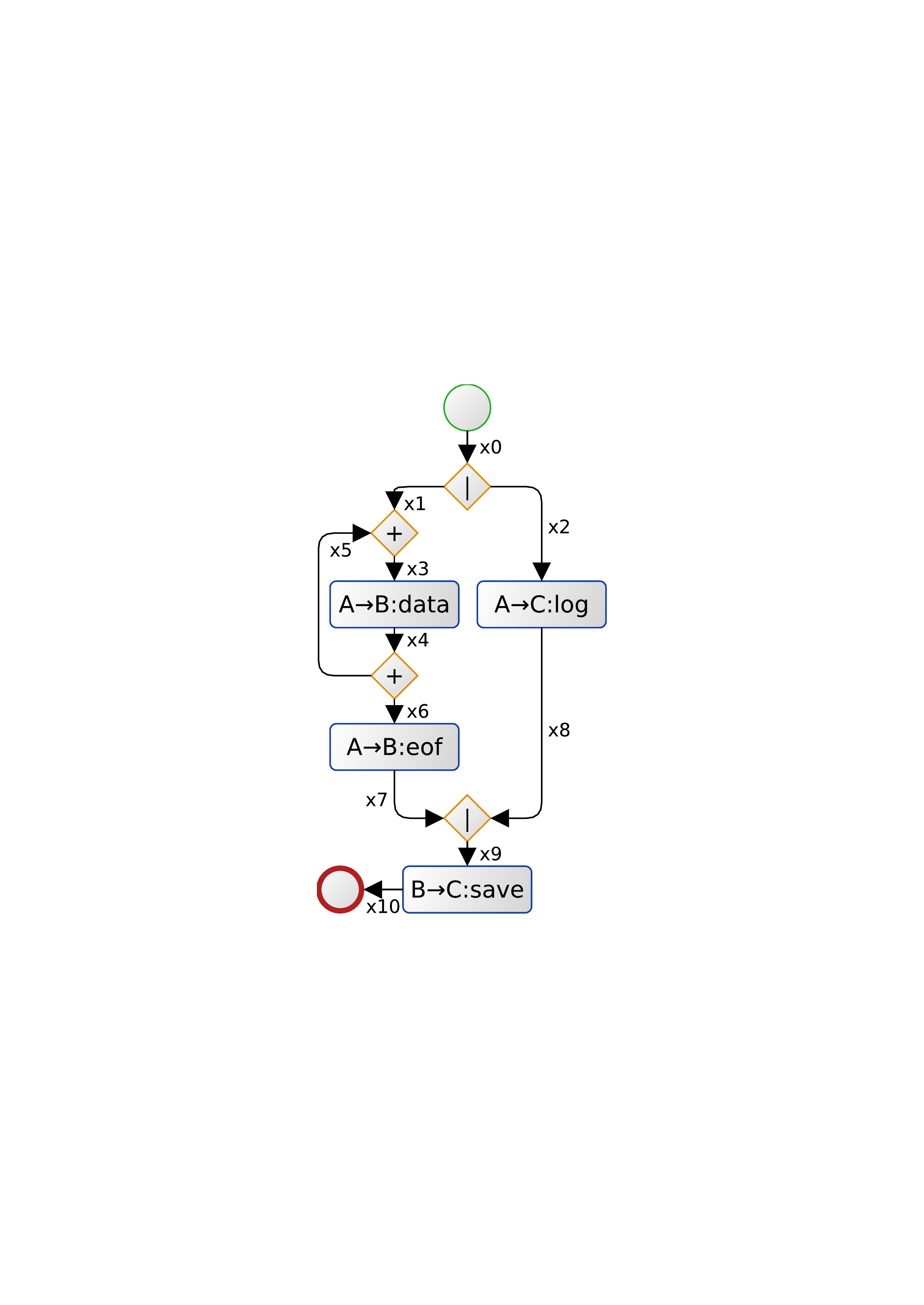}
    \end{minipage}
\end{array}
\]
\caption{Generalised global type and graph representation}
\label{fig:data-global}
\end{figure}

The prescribed interaction starts from $\xx_0$, which we call the {\em
  initial state} (in green in the graphical representation), and proceeds
according to the transitions specified in \GV\ (the diamond or boxes
operators in the picture). The {\em state variables} $\xx$ in $\GV$ (the
edges in the graph) represent the successive distributed states of the
interaction.  Transitions can be {\em message exchanges} of the form $\xx_3 =
\TO{A}{B}:\textit{data}\;;\xx_4$ where this transition specifies that $A$ can go from
$\xx_3$ to the continuation $\xx_4$ by sending message $\textit{data}$, while $B$ goes
from $\xx_3$ to $\xx_4$ by receiving it. In the graph, message exchanges
are represented by boxes with exactly one incoming and one outgoing edges.
$\xx_4 = \xx_5 + \xx_6$ represents the choice 
between continuing with $\xx_5$ or $\xx_6$ and $\xx_0 = \xx_1
\PAR \xx_2$ represents forking the interactions, allowing the
interleaving of actions at $\xx_1$ and $\xx_2$. These forking threads
are eventually collected by joining construct of the form $\xx_7 \PAR \xx_8 = \xx_9$.
Similarly choices (i.e. mutually exclusive paths) are closed by merging
construct $\xx_1 + \xx_5 = \xx_3$, where they share a continuation. Forks,
choices, joins and merges are represented by diamond ternary operators in
the graphical notation. Fork and choice have one input and two outputs,
join and merge have two inputs and one output. Fork and join use the
diamond operator with the $\PAR$ symbol, while choice and merge use a
diamond with the $+$ symbol. The $\xx_{10} = \End$ transition is
represented by a red circle.  Note that the two representations (syntax and
graph) are equivalent.

The motivation behind this choice of syntax is to support general control
flows, as classical global type syntax tree, even with added operators fork
$\PAR$ and choice $+$~\cite{CastagnaDP11,BettiniCDLDY08,CHY07,DY11}, is
limited to series-parallel control flow graphs.

\myparagraph{Generalised local types}
As for global types, a local type $\TT$ follows a shape of 
a state machine-like definition: local types are of the form $\defk\ \TV \ \In\
\xx_0$.  The different 
actions include send ($\pp!a$ is the action of sending
to $\pp$ a message $a$), 
receive ($\pp?a$ is the action of receiving 
from $\pp$ a message $a$), fork, internal choice,
external choice, join, merge, indirection and end.  Note that merge is used
for both internal and external choices. Similarly to global types, an obvious
graphical representation exists.

\iflong
\begin{figure}[t]
\fi
\iflong%
\mycaption{Generalised Local Types}
\label{fig:localtypes}
\fi%
{\small 
$$
\begin{array}{l}
\begin{array}{@{}r@{\,}c@{\,}ll@{\,}r@{\ }c@{\ }ll}
  \TT & \grmeq & \defk\ \TV \  \In\  \xx &
   \text{local type}  & \\
  T &\grmeq & \xx = \pp!a.{\xx'} & \text{send}&
  &\grmor & \xx = \xx' \oplus \xx'' & \text{internal choice}\\
  &\grmor & \xx = \pp?a.{\xx'} & \text{receive} & 
  &\grmor & \xx = \xx' \OR \xx'' & \text{external choice} \\ 
  &\grmor & \xx = \xx' \PAR \xx'' & \text{fork} &
  &\grmor & \xx + \xx' = \xx'' & \text{merge}\\
  &\grmor & \xx \PAR \xx' = \xx'' & \text{join} &
  &\grmor & \xx = \xx' & \text{indirection} \\
  &\grmor & \xx = \End & \text{end}
\end{array}
\end{array}
$$}
\iflong%
\hrule
\else%
\fi%
\iflong
\end{figure}
\fi

The local types are obtained 
from the global type by successive projection to
each participant. We define the
projection of a well-formed global type $\GG$ to the local type of participant
$\p$ (written $\GG\proj\p$). 
The projection is 
\ifpopllong
given in Appendix \ref{app:graph} 
\else
omitted 
\fi
because it is straightforward: 
for example, $\xx = \TO{p}{q}:a\;;\xx'$ is projected to 
the output $\xx = \SEND{\p'}{a}.\xx'$ 
from $\pp$'s viewpoint 
and an input 
$ \xx = \RECV{\p}{a}.\xx'$ 
from $\qq$'s viewpoint; 
otherwise it creates an indirection link from $\xx$ to
$\xx'$. Choice 
$\xx = \xx' + \xx''$ is projected to the internal choice 
$\xx = \xx' \oplus \xx''$ 
if 
$\pp$ is the unique participant deciding on which
branch to choose; otherwise the projection gives an external choice 
$\xx = \xx' \& \xx''$ (\cite{DY12} gives the definition).
Forks, joins and merges are kept identical. As an
example, Figure~\ref{fig:data-local} features on the left, in
graphical notation, the result of the projection to $A$ from the global
type $\GG$ of Figure~\ref{fig:data-global}. Its structure is exactly the
same as the original global type, except for the silent transition
$\xx_9=\xx_{10}$ which is silent from the point of view of $A$ and
therefore is just elided in the local type.

\subsection{Labelled transitions of generalised global and local types} %
\label{sec:lts-general}
It is possible to define a labelled semantics for global and local types by
considering the type (whether local or global) as a state machine
specification in which each participant (or the participant, in the case of
local type) can evolve, as they would in a CFSMs. As for CFSMs and
classical multiparty session types, we keep the syntax of labels ($\ell,
\ell', ...$).

We use the following notation to keep track of local states (with
parallelism, each participant can now execute several transitions concurrently):
\[
\begin{array}{l}
  \X  \grmeq  \xx_i  \grmor  \X \PAR \X\quad \quad 
  \X[\_]  \grmeq  \_  
  \grmor  \X[\_] \PAR \X 
  \grmor  \X \PAR \X[\_] 
\end{array}
\]


\myparagraph{LTS for global types} %
We first define, for a global type $\GG=\defk\ \GV \ \In\ \xx_0$, a
transition system $\defk\ \GV \ \In\ \XXV,\wv \TRANS{\ell} \defk\ \GV \
\In\ \XXV',\wv'$, where $\XXV$ and $\XXV'$ represents a vector recording
the state of each of the participants $\XXV=\ASET{\XX_\p}_{\p\in\PSet}$ and
where $\wv$ represents the content of the communication buffers
$\ASET{w_{\q\q'}}_{\q\q'\in\PSet}$. The states for the global type
$\G=\defk\ \GV \ \In\ \xx_0$ are equipped with an equivalence relation
$\equivGV$, 
\ifpopllong
defined in Appendix \ref{app:graph:equiv},
\else
omitted here, 
\fi
which covers associativity, commutativity, forks
and joins, choices and merges. Initially,
$\XXV_0=\ASET{\xx_0}_{\p\in\PSet}$ and $\wv_0=\ASET{\NUL}_{\q\q'\in\PSet}$.
The LTS for global types is defined in Figure~\ref{fig:globallts}.

The semantics of global types, as given by the rules~\tsrule{GGR1,2},
follows the intuition of communicating systems: if the global type allows,
a participant at the right state can put a value in a communication buffer
and progress to the next state (\tsrule{GGR1}) or, if a value can be read,
a participant at the right state can consume it and proceed
(\tsrule{GGR2}). Rule~\tsrule{GGR3} allows participants that are not
concerned by a transition to go there for free. Fork, join, choice and
merge transitions are passed through silently by rule~\tsrule{GGR4}.

\begin{figure}[t]
\[
\begin{array}{c}
  \frac
  {\xx =  \TO{p}{p'}:a\;;\xx' \in \GV \quad \XX_\p = \XX[\xx] \quad w_{\p\p'}\in\wv
  }{\defk\ \GV \  \In\  \XXV,\wv\TRANS{\p\p'!a}\defk\ \GV \  \In\
    \XXV[\XX_\p\isnow\XX[\xx']],\wv[w_{\p\p'}\isnow w_{\p\p'}\cdot 
a]} \tsrule{GGR1}
  \\[5mm]
  \frac
  { \xx =  \TO{p}{p'}:a\;;\xx' \in \GV \quad \XX_{\p'}=\XX[\xx] \quad
    w_{\p\p'}\in\wv \quad w_{\p\p'}=a \cdot w_{\p\p'}'}
  {\defk\ \GV \  \In\ \XXV,\wv \TRANS{\p\p'?a}\defk\ \GV \ \In
    \ \XXV[\XX_{\p'}\isnow\XX[\xx']],\wv[w_{\p\p'}\isnow w_{\p\p'}']} \tsrule{GGR2}
  \\[5mm]
  \frac
  {\scriptsize\begin{array}{@{}c@{}}
      \xx =  \TO{p}{p'}:a\;;\xx' \in \GV \quad \XX_{\q}=\XX[\xx] \quad
    \q\not\in\ASET{\p,\p'}\\
    \defk\ \GV \  \In\  \XXV[\XX_{\q}\isnow \XX[\xx']],\wv \TRANS{\ell}
    \defk\ \GV \  \In\  \XXV',\wv' 
  \end{array}} 
  {\defk\ \GV \  \In\  \XXV,\wv \TRANS{\ell}\defk\ \GV \  \In\
    \XXV',\wv'} \tsrule{GGR3}
  \\[3mm]
  \frac
  {    \XX_{\p}=\XX \quad \XX\equivGV\XX' \quad
    \defk\ \GV \  \In\  \XXV[\XX_{\p}\isnow \XX'],\wv \TRANS{\ell} \defk
    \ \GV \  \In\  \XXV',\wv'} 
  {\defk\ \GV \  \In\  \XXV,\wv \TRANS{\ell}\defk\ \GV \  \In
    \ \XXV',\wv'} \tsrule{GGR4}
\end{array}
\]
\caption{Global LTS}\label{fig:globallts}
\end{figure}

\myparagraph{LTS for local types} %
We define in Figure~\ref{fig:locallts} a transition system $\TTV,\wv
\TRANS{\ell} \TTV',\wv'$, where $\TTV$ represents a set of local types
$\ASET{\defk\ \TV \ \In\ \XXV_\p}_{\p\in\PSet}$ and $\wv$ represents the content of the
communication buffers $\ASET{w_{\q\q'}}_{\q\q'\in\PSet}$. Initially,
$\TTV_0$ sets all the local types to $\xx_0$ and
$\wv_0=\ASET{\NUL}_{\q\q'\in\PSet}$. The principle is strictly identical to
the LTS for global types, with, again, an omitted structural equivalence
$\equivTV$ between local states.

\begin{figure}[t]
  \[ \begin{array}{c}
    \frac
    {\xx = \pp'!a.{\xx'}\in \TV \quad \TT_\p=\defk\ \TV \ \In\  \X[\xx] \quad w_{\p\p'}\in\wv
    }{\TTV,\wv
      \TRANS{\p\p'!a}\TTV[\TT_\p\isnow\defk\ \TV \ \In\ \X[\xx']],\wv[w_{\p\p'}\isnow
      w_{\p\p'}\cdot \p\p'!a]} \tsrule{GLR1}
    \\[5mm]
    \frac
    { \xx =  \pp'?a.\xx' \in \TV \quad \TT_{\p'}=\defk\ \TV \  \In\ \XX[\xx] \quad
      w_{\p\p'}\in\wv \quad w_{\p\p'}=\p\p'!a \cdot w_{\p\p'}'}
    {\TTV,\wv \TRANS{\p\p'?a}\TTV[\TT_{\p'}\isnow \defk\ \TV \ \In
      \ \XX[\xx']],\wv[w_{\p\p'}\isnow w_{\p\p'}']}  \tsrule{GLR2}
    \\[3mm]
  \frac
  {    \TT_\p=\defk\ \TV \ \In\  \XX \quad \XX\equivTV\XX' \quad 
    \TV[\TT_\p\isnow \defk\ \TV \  \In\  \XX'],\wv \TRANS{\ell} \TTV',\wv'} 
  {\TV,\wv \TRANS{\ell}\TV',\wv'}  \tsrule{GLR3}
  \end{array}
  \]
\caption{Local LTS}\label{fig:locallts}
\end{figure}

\myparagraph{Equivalence between generalised local and global types}

Given the similarity in principle between the global and local LTSs, and
considering that the projection algorithm for generalised global types is
quasi-homomorphic, we can easily get the trace equivalence between the
local and global semantics.

\begin{theorem}[soundness and completeness of projection]
\label{thm:general:lts}
If $\vec{\TT}$ is the projection of a global type $\GG$ to all roles, then
$\GG \WB (\vec{\TT},\NUL)$.
\end{theorem}

\subsection{Translations between general local types and CFSMs}
\label{sec:translationtolocal}
Now that we have proved the equivalence from global to local types, we
establish the conversion of local types to and from CFSMs.

\myparagraph{Translation to CFSMs} %

We first give the already known translation from local types to
CFSMs~\cite{DY12}. The illustration of that translation on the Data
transfer example is given on the top-right corner of
Figure~\ref{fig:data-local}.

\begin{definition}[translation from local types to CFSMs~\cite{DY12}]
\label{def:graph:translation}
\rm If $\TT = \defk\ \TV \ \In\ \xx_0$ is the local type of participant $\pp$
  projected from $\GG$, then the corresponding automaton is
  $\A(\TT)=(Q,C,q_0,{\Bbb A},\delta)$ where:
\begin{itemize}
\item $Q$ is defined as 
the set of well-formed states $\X$
built from the state variables $\{ \xx_i \}$ of $\TT$.  
%
%
$Q$ is defined up to the equivalence relation $\equivTV$ mentioned
in \S~\ref{sec:lts-general}.


\item $C = \{ \p\q \PAR \p,\q \in \GG\}$
\item $q_0= \xx_0$
\item $\Sigma$ is the set of $\{a \in \GG\}$ 
\item $\delta$ is defined by:
  \begin{itemize}
  \item $(\X[\xx],(\p\p'!a),\X[\xx'])\in \delta$ if $\xx =
    \p'!a.{\xx'} \in \TV$.
  \item $(\X[\xx],(\p'\p?a),\X[\xx'])\in \delta$ if $\xx = \pp'?a.{\xx'}
    \in \TV$. 
  \end{itemize}
\end{itemize}
\end{definition}

\begin{figure}

\[
\begin{array}{@{}l@{\hspace{-2em}}l@{}}
    \begin{minipage}{12.5em}
      \includegraphics[trim=6cm 9cm 5cm 11cm,width=12.5em]{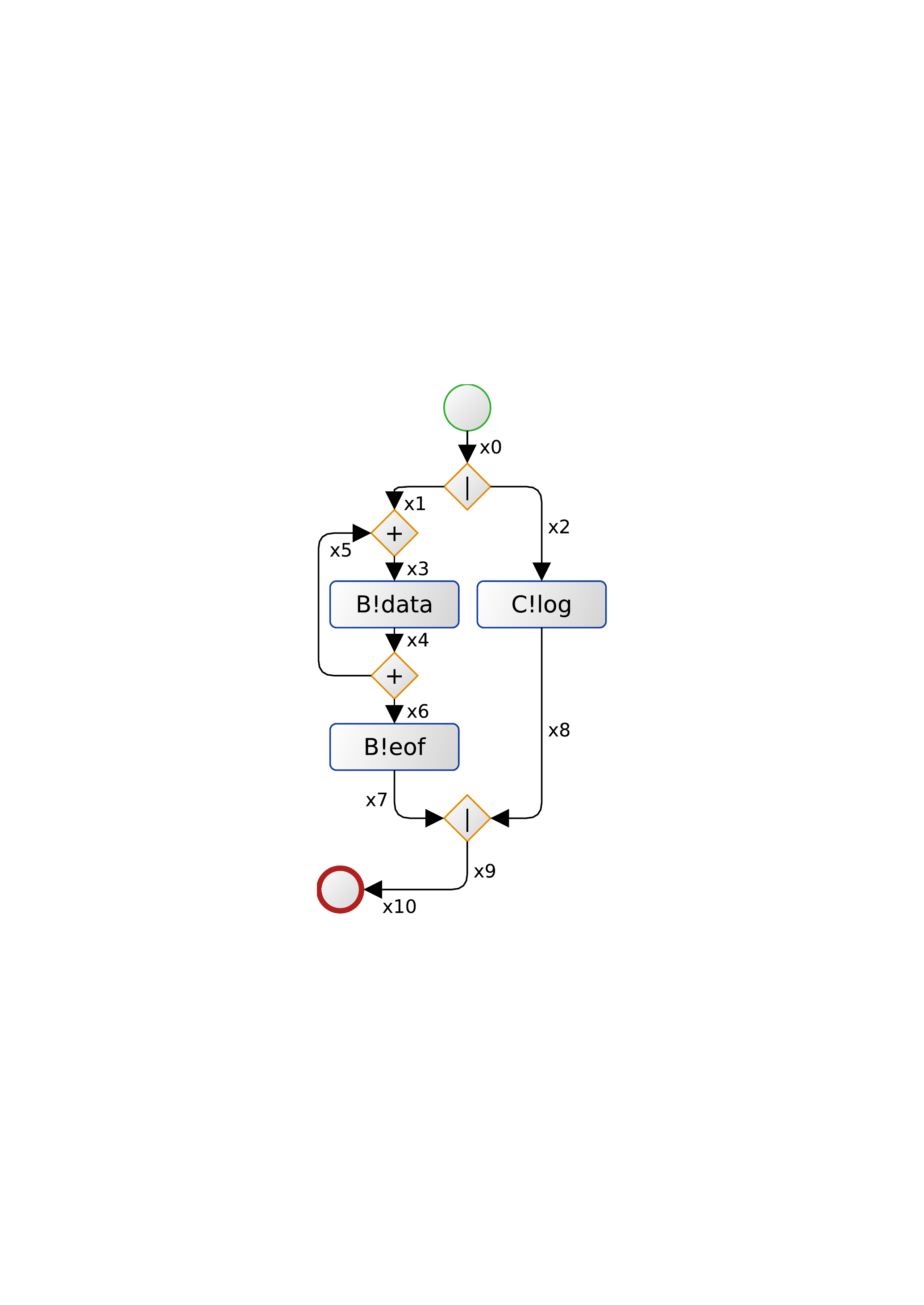}

      \smallskip
      \begin{center}
         General local type for A
       \end{center}
     \end{minipage} &
    \begin{array}{@{}c@{}}
      \xymatrix@C=4em@R=4ex{
        *++[o][F=]{} \ar[r]_{AC!\text{log}} \ar[d]_{AB!\text{data}}
        & *++[o][F]{}\ar[d]_{AB!\text{data}}\\
        *++[o][F]{}\ar[d]_{AB!\text{eof}}\ar[r]_{AC!\text{log}}  \ar@(lu,ld)[]_{AB!\text{data}}
        & *++[o][F]{}\ar[d]_{AB!\text{eof}}  \ar@(ru,rd)[]^{AB!\text{data}}\\
        *++[o][F]{}\ar[r]_{AC!\text{log}} 
        & *++[o][F]{} \\
      }\\[1.5ex]
      \text{CFSM}\\
    \begin{minipage}{12.5em}
      \includegraphics[trim=6.5cm 12cm 6.5cm 11.5cm,width=12.5em]{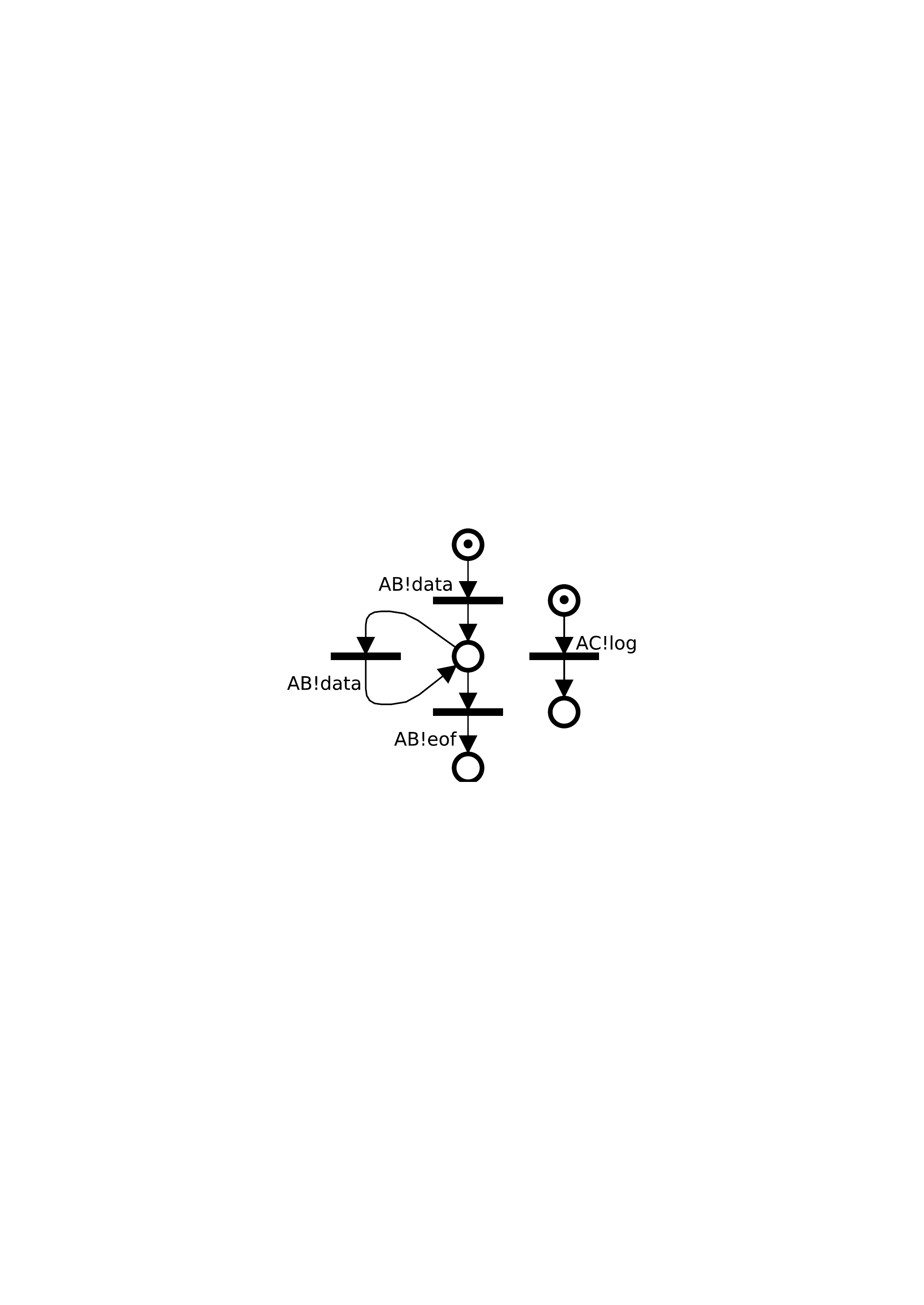}

      \smallskip
      \begin{center}
         Inferred labelled Petri net
       \end{center}
    \end{minipage}
  \end{array}
\end{array}
\]
\caption{Data transfer example: local translations}\label{fig:data-local}
\end{figure}

\myparagraph{Translations from CFSMs} %
The converse translation is not as obvious as local types feature explicit
forks and joins, while CFSMs only propose choices between interleaved
sequences. The translation from a CFSM to a local type therefore comes in 3
steps.

First, we apply a generic translation from minimised CFSMs to Petri
nets~\cite{nielsen1992elementary,cortadella1998deriving}. This translation
relies on the polynomial computation of the graph of
regions~\cite{badouel1998theory}, preserves the trace semantics of the CFSM
and, by the minimality of the produced net, makes the concurrency
explicit. 
Figure~\ref{fig:data-local} illustrates on the Data transfer example the shape
of the Petri net that can be produced by such a generic translation. Note
that the produced Petri net is always safe and free choice.

The second step of the conversion is to take the Petri net with labelled
transitions and enrich it with new silent transitions and new places so
that it can be translated into local types. Notably, it should have only
one initial marked place, one final place and all labelled transitions
should have exactly one incoming and one outgoing arc. Then, we constrain
all transitions to be linked with no more than 3 arcs (2 incoming and 1
outgoing for a join transition, or 1 incoming and 2 outgoing for a fork
transition, 1 incoming and 1 outgoing for all the other
transitions). Places should have no more that 2 incoming and 2 outgoing
arcs: if there are two incoming (merge), then the transitions they come
from should only have one incoming arc each; if there are 2 outgoing
(choice), then the transitions they lead to should have only one outgoing
arc each.

In the end, the translation to local type is simple, as each place
corresponds to a state variable $\xx$, and the different local type
transitions can be simply identified. For the lightness of the
presentation, instead of defining formally this last step, we describe the
converse translation. From it, it is possible to infer the local type
generation. 

\begin{definition}[Petri net representation]\label{def:petri-net-local}\rm
Given a local type $\TT=\defk\ \TV \ \In\ \xx_0$, we define the Petri net
$\PN(\TT)$ by:
\begin{itemize}
\item Each state variable $\xx\in\TV$ is a place in $\PN(\TT)$.
\item All the places are initially empty, except for one token in $\xx_0$.
\item Transitions in $\TV$ are translated as follows:
  \begin{itemize}
  \item If $\xx=\p!a.\xx'\in \TV$ then there is a
    transition labelled in $\PN(\TT)$, whose unique input arc comes from $\xx$ and
    whose unique output arc goes to $\xx'$.
  \item If $\xx= \p?a.\xx'\in \TV$ then their is a
    transition in $\PN(\TT)$, whose unique input arc comes from $\xx$ and
    whose unique output arc goes to $\xx'$.
  \item If $\xx_1=\xx_2\PAR\xx_3\in \TV$ then there is a transition in
    $\PN(\TT)$, whose unique input arc comes from $\xx_1$ and whose two
    outputs arcs go to $\xx_2$ and $\xx_3$.
  \item If $\xx_1=\xx_2+\xx_3\in \TV$ (internal or external choice) then
    there are two transitions in $\PN(\TT)$, that each have an input arc
    from $\xx_1$ and that respectively have an output arc to $\xx_2$ and
    $\xx_3$. 
  \item If $\xx_1+\xx_2=\xx_3\in \TV$ then there are two transitions in
    $\PN(\TT)$, that respectively have an input arc from $\xx_1$ and $\xx_2$
    and that both have an output arc to $\xx_3$.
  \item If $\xx_1\PAR\xx_2=\xx_3\in \TV$ then there is a transition in
    $\PN(\TT)$, whose two input arcs respectively come from $\xx_1$ and
    $\xx_2$ and whose unique output arc goes to $\xx_3$.
  \end{itemize}
\end{itemize}
\end{definition}
The idea of the translation back from a Petri net to a local type is
to identify the transitions and place patterns and convert them into local
type transitions.

Note that, in Figure~\ref{fig:data-local}, the inferred Petri Net will not
give back the local type on the left: in the general case, going through
the translation from local type to CFSM and then back to local type will
only give an isomorphic local type. The traces are of course preserved.



\subsection{Parallelism and local choice condition}
\label{subsec:par_localchoice}
This subsection introduces the conditions that CFSMs should respect in
order to correspond to well-formed local types projected from generalised
global types. It extends the conditions that were sufficient for classical
multiparty session types for two reasons. First, we now have concurrent
interactions and the no-mixed choice condition does not hold
anymore. Second, the well-formedness condition corresponding to
projectability in classical multiparty session types needs to take into
account the complex control flows of generalised multiparty session types. 


We start by a commutativity condition for mixed states in CFSMs: a state is
mixed parallel if any send transition satisfies the diamond property with
any receive transition. Formally:
 
\begin{definition}[mixed parallel]
  \rm Let $M=(Q,\chanset,q_0,\ASigma,\delta)$.  We say local state $q$ in
  $M$ is {\em mixed parallel} if for all
  $(q,\ell_1,q_1'),(q,\ell_2,q_2')\in \delta$ such that $\ell_1$ is a send
  and $\ell_2$ is a receive we have $(q_1',\ell_2,q'),(q_2',\ell_1,q') \in
  \delta$ for some $q'$.
\end{definition}


Next, we introduce two conditions for the choice that are akin to the local
choice conditions with additional data of~\cite[Def.~2]{GenestMP03} or the
``knowledge of choice'' conditions of~\cite{CastagnaDP11}. 

\begin{definition}[local choice condition]\rm
\label{def:localchoice}
\begin{enumerate}
\item 
The set of {\em receivers} of transitions 
$s_1 \TRANSS{t_1\cdots t_m} s_{m+1}$ 
is defined as $\Rcv(t_1\cdots t_m)=\ASET{ \q \PAR 
\exists i\leq  m, t_i= (s_{i},\p\q?a,s_{i+1})}$. 
\item 
The set of {\em active senders} are defined as 
$\ASend(t_1\cdots t_m)=
\ASET{\p \PAR \exists i\leq m, t_i=
(s_{i},\p\q!a,s_{i+1}) \wedge \forall k<i. \ t_k \not =
(s_{k},\p'\p?b,\s_{k+1})}$ 
and represent the participants who could immediately send from state $s_1$.

\item Suppose $s_0 \TRANSS{\varphi} s$ and $\varphi=\varphi_0 \cdot t_1 \cdot 
\varphi_1 \cdot t_2 \cdot \varphi_2$. 
We write $t_1 \triangleleft t_2$ ($t_2$ depends on $t_1$) if 
either (1) $\Phi(\act(t_2))=\act(t_1)$ or (2) $\subj(t_1)=\subj(t_2)$ 
unless $t_1$ and $t_2$ are parallel. 

\item We say $\varphi=t_0\cdot t_1 \cdot t_2 \cdots t_n$ is {\em the causal 
chain} if $s_0 \TRANSS{\varphi'} s'$ and $\varphi\subseteq \varphi'$ 
with, for all $0 \leq k \leq n-1$, there exists $i$ such that $i> k$ 
and $t_k \triangleleft t_i$.    


\item $S$ satisfies the {\em receiver property} if,   
for all $s\in \RS(S)$ and $s\TRANSS{t_1} s_1$ and 
$s\TRANSS{t_2} s_2$ with $\act(t_i) = \p\q_i!a_i$,  
there exist $s_1 \TRANSS{\varphi_1}s_1'$ and $s_2 \TRANSS{\varphi_2}s_2'$  
such that $\Rcv(\varphi_1)=\Rcv(\varphi_2)$.

\item $S$ satisfies the {\em unique sender property}
if $s_0 \TRANSS{\varphi_1}s_1 \TRANSS{t_1} s_1'$ and 
$s_0 \TRANSS{\varphi_2}s_2 \TRANSS{t_2} s_2'$, 
with $\act(t_1)=\p_1\p?a_1$, $\act(t_2)=\p_2\p?a_2$, $a_1 \not
= a_2$, $\neg t_1 \triangleleft t_2$ and $\neg t_2 \triangleleft t_1$, 
and $\varphi_i\cdot t_i$ the maximum causal chain. 
Then $\ASend(\varphi_1\cdot t_1)=\ASend(\varphi_2\cdot t_2) = \ASET{\q}$.
\end{enumerate}
\end{definition}
Together with multiparty compatibility, the receiver property 
ensures deadlock-freedom while the unique sender property 
guarantees orphan message-freedom. 

\begin{proposition}[stability]
\label{thm:graph:stable}
  Suppose $S=\ASET{M_\p}_{\p \in \PSet}$ and each $M_\p$ is
deterministic. If (1)
$S$ is multiparty compatible; (2) each mixed state in $S$ is mixed
parallel; and (3) for any local state that can do two receive
transitions, either they commute (satisfy the diamond property) or 
the state satisfies the unique sender condition, then $S$ is stable and
satisfies the reception 
error freedom and orphan message-freedom properties. 
\end{proposition}
\begin{proof}
The proof is similar to Proposition \ref{pro:core:stable}, noting that 
the unique sender condition guarantees the input availability. 
See 
\ifpopllong
Appendix \ref{app:graph}.  
\else 
\cite{fullversion}.
\fi
\QED 
\end{proof}

\begin{theorem}[deadlock-freedom]
\label{thm:graph:df}
Suppose $S=\ASET{M_\p}_{\p \in \PSet}$ 
satisfies the 
same conditions as Proposition \ref{thm:graph:stable}. 
Assume, in addition, that $S$ satisfies the receiver condition. 
Then $S$ is deadlock-free. 
\end{theorem}
\begin{proof}
  We deduce this theorem from the stability property and the receiver
  condition. The proof uses a similar reasoning
  as Proposition~\ref{pro:cmsa:safety}. \QED \\
\end{proof}

We call the systems that satisfy the conditions of
Theorem~\ref{thm:graph:df} {\em session-compatible}.

By the same algorithm, the multiparty compatibility property is decidable
for systems of deterministic CFSMs. It is however undecidable to check the
receiver and unique sender properties in general. On the other hand, once
multiparty compatibility is assumed, we can restrict the checks to
1-bounded executions (i.e.~we limit $\varphi_1$, $\varphi_2$, $\varphi_1'$
and $\varphi_2'$ to 1-bounded executions and $\RS_1(S)$ in Definition
\ref{def:localchoice}). Then these properties become decidable.
Combining the synthesis algorithm defined below, we can decide 
a subset of CFSMs which can build a general, well-formed global type.

\subsection{Synthesis of general multiparty session automata}
Now all the pieces are in place for the main results of this paper. We are
able to identify the class of communicating systems that correspond to
generalised multiparty session types.

The main theorems in this section follow:

\begin{theorem}[synthesis of general systems] \label{thm:gen:syn}
Suppose $S=\ASET{M_\p}_{\p\in \PSet}$ is a session-compatible system.
Then there is an algorithm which builds $\GG$ such that $S \WB \GG$. 
\end{theorem}
\begin{proof}
  The algorithms is the following. We consider $S=\ASET{M_\p}_{\p\in
    \PSet}$ as the definition of a transition system. In this transition
  system, we only consider the 1-bounded executions. 
  This restriction produces a
  finite state LTS, where send transitions are immediately followed by the
  unique corresponding receive transition. 
  In each of these cases, we
  replace the pair of transitions $\p\p'!a$ and $\p\p'?a$ by a unique
  transition $\TO{p}{p'}:a$. To obtain the global type $\GG$, we then
  follow first the standard conversion to Petri nets and the equivalence
  between Petri nets and global types (similar to the one between Petri
  nets and local types). We conclude the equivalence by a
  version of Lemma~\ref{lem:one_to_k} adapted to session-compatible system. \QED
\end{proof}

Using the synthesis theorem, we are able to provide a full
characterisation of generalised multiparty session types in term of
session-compatible systems.

\begin{theorem}[soundness and completeness in \MSA] 
Suppose $S=\ASET{M_\p}_{\p\in \PSet}$ is a session compatible system.
Then there exits $\GG$ such that 
$S \WB \GG$. 
Conversely, if $\GG$ is well-formed as in \cite{DY12},  
then there exits $S$ which satisfies the safety and liveness properties 
(deadlock-freedom, reception error-freedom and orphan message-freedom),
and $S \WB \GG$.  
\end{theorem}
\begin{proof}
  By Theorem~\ref{thm:gen:syn} and Theorem \ref{thm:general:lts} with the
  same reasoning as in Theorem \ref{thm:soundcompletemsacore}. \QED
\end{proof}

\section{Appendix for Section \ref{sec:graph}}
\label{app:graph}

\paragraph{Projection}
We define the projection from a global type to a local type  
where $\mathit{ASend}$ means that 
a set of active senders, which corresponds to the same definition 
in CFSMs (see \cite{DY12}). 

\iflong
\begin{figure}[t]
\fi
\iflong%
\mycaption{Projection Algorithm}
\label{fig:projection}
\fi%
{\small 
$$
\begin{array}{r@{\ }l@{\quad}c@{\quad}l@{\hspace{-2em}}l}
  \defk\  \GV\  \In\  \xx  & \proj \p & = & 
 \defk\ \GV \proj_\GV \p\  \In\  \xx \\[1mm]
 \xx = \TO{p}{p'}:a\;;\xx' &  \proj_\GV \, \p & = &
 \xx = \SEND{\p'}{a}.\xx'\\
 \xx = \TO{\p}{\p'}:a\;;\xx' &  \proj_\GV \, \p' & = &
 \xx = \RECV{\p}{a}.\xx'\\
 \xx = \TO{p}{p'}:a\;;\xx' &  \proj_\GV \, \p'' & = &
 \xx = \xx'\ (\pp\notin\{\p,\p'\})\\
 \xx \PAR \xx' = \xx'' & \proj_\GV \, \p & = & 
 \xx \PAR \xx' = \xx'' \\
 \xx = \xx' \PAR \xx'' & \proj_\GV \, \p & = & 
 \xx = \xx' \PAR \xx'' \\
 \xx = \xx' + \xx'' & \proj_\GV \, \p & = & 
 \xx = \xx' \oplus  \xx'' 
&  \text{(if $\p=\Asender(\GV)(\xx)$)}\\
 \xx = \xx' + \xx'' & \proj_\GV \, \p & = & 
 \xx = \xx' \OR \xx''
& \text{(otherwise)}
  \\
  \xx + \xx' = \xx'' & \proj_\GV \, \p & = & 
  \xx + \xx' = \xx'' \\
  \xx =  \End & \proj_\GV \, \p & = & \xx =  \End
\end{array}
$$}
\iflong%
\hrule
\else%
\fi%
\iflong
\end{figure}
\fi

\subsection{Global type equivalence}
\label{app:graph:equiv}
Below we define the equivalence relation $\equivGV$ used in the
LTS of the global types. 

\iflong%
\begin{figure}[t]
\mycaption{Global State Equivalence}
\label{fig:globalstateequivalence}
\else
\fi
\[
\small
\begin{array}{@{}c}
\X\PAR\X'\equivGV\X'\PAR\X \qquad 
\X\PAR(\X'\PAR\X'')\equivGV(\X\PAR\X')\PAR\X''\\[1mm]
\infer{\xx = \xx' \in \GV}{ \X[\xx]\equivGV\X[\xx']}\quad 
\infer{\xx = \xx' \PAR \xx'' \in \GV}{ \X[\xx]\equivGV\X[\xx' \PAR
  \xx'']} \quad 
\infer{\xx \PAR \xx' = \xx'' \in \GV}{
  \X[\xx\PAR\xx']\equivGV\X[\xx'']}
\\[5mm] 
\infer{\xx = \xx' + \xx'' \in \GV}{ \X[\xx]\equivGV\X[\xx']}\quad 
\infer{\xx = \xx' + \xx'' \in \GV}{ \X[\xx]\equivGV\X[\xx'']}\quad
\infer{\xx + \xx' = \xx'' \in \GV}{ \X[\xx]\equivGV\X[\xx'']}\quad 
\infer{\xx + \xx' = \xx'' \in \GV}{ \X[\xx']\equivGV\X[\xx'']}
\end{array}
\]
\iflong%
\vspace{3ex}
\hrule%
\end{figure}
\else%
\fi

Below we define the equivalence relation $\equivTV$ used in the
translation in Definition \ref{def:graph:translation}.  

\iflong%
\begin{figure}[t]
\mycaption{Local State Equivalence for Local State Automata}
\label{fig:localstateequivalenceforlocal}
\else
\begin{centering}
\fi
\small
$
\begin{array}{@{}c}
\X\PAR\X'\equivTV\X'\PAR\X \qquad 
\X\PAR(\X'\PAR\X'')\equivTV(\X\PAR\X')\PAR\X''\\[1mm]
\infer{\xx = \xx' \in \TV}{ \X[\xx]\equivTV\X[\xx']}\quad 
\infer{\xx = \xx' \PAR \xx'' \in \TV}{ \X[\xx]\equivTV\X[\xx' \PAR
  \xx'']} \quad 
\infer{\xx \PAR \xx' = \xx'' \in \TV}{
  \X[\xx\PAR\xx']\equivTV\X[\xx'']}
\\[5mm] 
\infer{\xx = \xx' \OR \xx'' \in \TV}{ \X[\xx]\equivTV\X[\xx']}\quad 
\infer{\xx = \xx' \OR \xx'' \in \TV}{ \X[\xx]\equivTV\X[\xx'']}\quad 
\infer{\xx = \xx' \oplus \xx'' \in \TV}{ \X[\xx]\equivTV\X[\xx']}\quad
\infer{\xx = \xx' \oplus \xx'' \in \TV}{
  \X[\xx]\equivTV\X[\xx'']}\\[5mm]
\infer{\xx + \xx' = \xx'' \in \TV}{ \X[\xx]\equivTV\X[\xx'']}\quad 
\infer{\xx + \xx' = \xx'' \in \TV}{ \X[\xx']\equivTV\X[\xx'']}
\end{array}
$
\iflong%
\vspace{3ex}
\hrule%
\end{figure}
\else%
\end{centering}
\fi

\subsection{Proof of Proposition~\ref{thm:graph:stable}}
\label{app:graph:stable}
Essentially we have the same as the proof of Proposition
\ref{pro:core:stable}. Only difference is that we need to use 
the unique sender condition to ensure that 
the action $\overline{t_1}$ is possible 
in (\ref{proof:inputavailability}) in the proof of Proposition
\ref{pro:core:stable} (note that $\overline{t_1}$ is always 
possible in basic CFSMs since they are directed). 

Suppose, in
(\ref{proof:inputavailability}) in the proof of Proposition
\ref{pro:core:stable}, the action $\overline{t_1}$ is not
possible: i.e.
$s_1 \TRANSS{t_1} s_2 \TRANSS{\varphi_2}s_2'$
but $s_2'$ cannot perform $\TRANSS{\overline{t_1}}$. The only
possibility is that some $M_\q$ contains the receiver state $q$ such
that $(q,\p\q?a,q'),(q,\p'\q?b,q'')\in \delta_\q$ which does not
satisfy the parallel condition (since if so, $s_2'$ can perform $\TRANSS{\overline{t_1}}$), and $\varphi_2$
contains the action $\p'\q?b$, which implies $\varphi_2$ contains 
the action $\p'\q!b$. By the unique sender condition,
there is the unique $\q'$ such that 
$s_0'\TRANSS{\varphi\cdot\p\q!a} s_1$ and 
$s_0'\TRANSS{\varphi\cdot\p\q!a\cdot \varphi'\cdot \p'\q!b \cdot \varphi''} s_2'$
with $\ASend(\varphi\cdot\p\q!a)=
\ASend(\varphi\cdot\p\q!a\cdot\varphi'\cdot \p'\q!b) = \ASET{\q'}$. Since 
$\p'\q!b$ cannot be reordered before $\p\q!a$ or after $\varphi_1$, 
to satisfy the unique sender property, $\varphi'$ should include
$\p\q?a$. This contradicts that the assumption that $\varphi_2$ does not include
$\p\q?a$. 

\subsection{Proof of Theorem~\ref{thm:graph:df}}
\label{app:graph:df}
By (reception error freedom) and (orphan message-freedom), together 
with (stable-property), we only have to
check, there is no input is waiting with an empty queue forever.  
Suppose by contradiction, there is $s\in \RS(S)$ such that 
$s=(\vec{q};\vec{\NUL})$ and there exists input state $q_\p\in
\vec{q}$ and no output transition from $q_k$ such that $k\not = \q$.     

Then by assumption, there is a 1-buffer execution $\varphi$ and  
since $\varphi$ is not taken (if so, 
$q_\p$ can perform an input), 
then there is another execution $\varphi'$ such that 
it leads to state $s$ which is deadlock at $q_\p$. \\[2mm]
{\bf Case (1)}
Suppose $\varphi$ does not include input actions at $\q$ except $a$,
i.e. $a$ is the first input action at $\q$ in $\varphi$. 
We let $\varphi_0$ for the prefix before the actions of $\q\p!a \cdot
\q\p?a$.  

By (receiver condition), we know $\p\in \Rcv(\varphi')$. 

By the determinacy, the corresponding input action has a different
label from $a$, i.e. $\q'\p?a'\in \varphi'$. By the diamond 
property,  
$\q'\p?a'$ and $\q\p?a$ can be appeared from the same state, i.e. this
state is under the assumption of the parallel condition. 
Hence by the multiparty compatibility, the both corresponding outputs 
$\q'\p!a'$ and $\q\p!a$ can be always 
fired if one of them is. This contradicts
the assumption that $q_\p$ is deadlock with label $a$.

\smallskip

\noindent{\bf Case (2)}
Suppose $\varphi$ includes other input actions at $\q$ 
before $\q\p?a$, i.e. $\p\in \Rcv(\varphi_0)$. 
Let $\q'\p?a'$ the action which first occurs in $\varphi_0$. 
By $\p\in \Rcv(\varphi')$, there exists $\q''\p?a''  \in \varphi'$. 
If $\q''\p?a'' \not= \q'\p?a'$, by the same reasoning as (1), 
the both corresponding outputs are available. Hence we assume 
the case $\q''\p?a'' = \q'\p?a'$. Let $s$ is the first state 
from which a transition in $\varphi_0$ and a transition 
in $\varphi'$ are separated. 
Then by assumption, 
if $s \TRANS{\varphi_0\cdot\q'\p!a'\cdot\q'\p?a'} s_1$ 
and $s \TRANS{\varphi_1\cdot\q'\p!a'\cdot\q'\p?a'} s_2$, 
by assumption  
$a'\not\in \varphi_0\cup \varphi_1$, 
hence $s \TRANS{\q'\p!a'\cdot\q'\p?a'}s_1'\TRANS{\varphi_0'} s_1$ 
and $s \TRANS{\q'\p!a'\cdot\q'\p?a'}s_2'\TRANS{\varphi_1'} s_2$
by the diamond property again. 
Since $s_1$ can perform an input at $\q$ by the assumption (because of
$\q\p?a$), $\varphi_1'$ should contain an input at $\q$ by the
receiver condition. If it contains the 
input to $\q$ in $\varphi_1'$, then  
we repeat Case (2) noting that the length of $\varphi_1'$ is shorter
than the length of  $\varphi_1\cdot\q'\p!a'\cdot\q'\p?a'$; else we use Case (1) to lead the contradiction; 
otherwise   if it contains the same input as $\q\p?a$, then it contradicts 
the assumption that $q_\p$ is deadlock.

\fi

\iflong
\makeatletter
\@starttoc{xmp}
\makeatother
\addtocontents{xmp}{\protect\setcounter{tocdepth}{3}}


\fi

\end{document}